\documentclass{article}

\usepackage{arxiv}

\usepackage{amsmath}		
\usepackage{amssymb}		
\usepackage{amsfonts}		
\usepackage{amsthm}		
\usepackage{mathtools}		
\usepackage{bbm}
\usepackage{bm} \newcommand\bell{\bm{\ell}} 
\usepackage{subcaption}
\usepackage{enumerate}
\usepackage[shortlabels]{enumitem}
\usepackage{xcolor}
\usepackage{placeins}
\usepackage{layouts}
\theoremstyle{definition}
\newtheorem{assumption}{Assumption}		

\newtheorem*{definition*}{Definition}		
\newtheorem*{assumption*}{Assumptions}		
\theoremstyle{plain}
\newtheorem{theorem}{Theorem}		
\newtheorem{lemma}{Lemma}		
\newtheorem{proposition}{Proposition}		

\newtheorem*{corollary*}{Corollary}		

\theoremstyle{remark}
\newtheorem*{remark*}{Remark}

\DeclareMathOperator*{\argmin}{arg\,min}		
\DeclareMathOperator{\trace}{trace}

\newcommand{\R}{\mathbb{R}}		

\newcommand{\debug}[1]{#1}		

\newcommand{\blue}[1]{{#1}}

\newcommand{\newmacro}[2]{\newcommand{#1}{\debug{#2}}}		

\newmacro{\ball}{\mathbb{B}}		
\newcommand{\acks}[1]{\section*{Acknowledgments}#1}

\newmacro{\play}{i}		
\newmacro{\game}{\mathcal{G}}		
\newmacro{\stratbeqci}{\stratb^{\ci*}}
\newmacro{\strateqci}{\strat^{\ci*}}
\newmacro{\potentci}{\potent^{\ci}}

\newcommand{\Cestim}{C_{\mathrm{estim}}} 
\newcommand{\Cestimols}{C^{\mathrm{OLS}}_{\mathrm{estim}}} 
\newcommand{\Cpublic}[1]{\Cestim(#1)} 
\newcommand\esp[1]{\mathbb{E}\left[#1\right]}     
\newcommand\p[1]{\left(#1\right)}                 
\newmacro{\ls}{\mathtt{LS}}
\newmacro{\gls}{\mathtt{GLS}}
\newmacro{\subplayset}{S_N}
\newmacro{\ols}{\mathtt{OLS}}
\newmacro{\mse}{\mathtt{MSE}}
\newmacro{\var}{\mathtt{Var}}
\newmacro{\bias}{\mathtt{Bias}}
\newmacro{\opt}{\ensuremath{\mathtt{opt}}}
\newmacro{\NE}{\mathtt{NE}}
\newmacro{\pos}{\mathtt{PoS}}
\newmacro{\poa}{\mathtt{PoA}}
\newmacro{\B}{\mathcal B}
\newmacro{\E}{\mathbb E}
\newmacro{\T}{T}
\newmacro{\prcs}{A}
\newmacro{\cov}{V}
\newmacro{\ndim}{d}
\newmacro{\dv}{{\boldsymbol d}}

\newmacro{\ud}{\textrm{d}}

\newmacro{\bigzero}{\makebox(0,0){\text{\huge0}}}
\newmacro{\SC}{\ensuremath{C_{\text{social}}}}
\newmacro{\Util}{J}

\newmacro{\ci}{\mathrm{ci}}
\newmacro{\Vm}{\bar{V}}
\newmacro{\epsr}{\varepsilon^{(r)}}
\newmacro{\etar}{\eta^{(r)}}
\newmacro{\Vr}{V^{(r)}}
\newmacro{\xb}{\boldsymbol{x}}
\newmacro{\yb}{\boldsymbol{y}}
\newmacro{\db}{\boldsymbol{d}}
\newmacro{\betab}{\boldsymbol{\beta}}
\newmacro{\betabh}{\hat{\boldsymbol{\beta}}}
\newmacro{\betah}{\hat{\beta}}
\newmacro{\psib}{\boldsymbol{\psi}}
\newmacro{\phib}{\boldsymbol{\phi}}
\newmacro{\epsilonb}{\boldsymbol{\varepsilon}}
\newmacro{\lambdab}{\boldsymbol{\lambda}}
\newmacro{\lambdamin}{\lambda_{\textrm{min}}}
\newmacro{\lambdamax}{\lambda_{\textrm{max}}}
\newmacro{\ellmax}{\ell_{\textrm{max}}}
\newmacro{\sigmamax}{\sigma_{\textrm{max}}}
\newmacro{\Vmax}{V_{\textrm{max}}}
\newmacro{\lambdaeq}{\lambda^{*}}
\newmacro{\lambdabeq}{\lambdab^{*}}
\newmacro{\Lambdaeq}{\Lambda^{*}}
\newmacro{\cp}{c^{\prime}}
\newmacro{\cpp}{c^{\prime\prime}}
\newmacro{\deltab}{\boldsymbol{\delta}}
\newmacro{\lambdabmax}{\lambdab_{max}}
\newmacro{\Exp}{E}
\newmacro{\Xset}{\mathcal{X}}
\newmacro{\yt}{\Tilde{y}}
\newmacro{\ybt}{\tilde{\yb}}
\newmacro{\info}{M}

\newmacro{\intermediate}{g}
\newmacro{\majorprec}{\nbplay^{-\frac{1 - \pmin}{1 - \pmax}\frac{\qhomo + 1}{\pmin + \qhomo}}}

\newmacro{\playsum}{j}
\newmacro{\nbplay}{n}
\newmacro{\playset}{N}
\newmacro{\types}{\mathcal{T}}

\newmacro{\phomo}{p}
\newmacro{\qhomo}{q}
\newmacro{\pmax}{{\phomo_{\max}}}
\newmacro{\pmin}{{\phomo_{\min}}}
\newmacro{\amax}{{a_{\max}}}
\newmacro{\amin}{{a_{\min}}}
\newmacro{\cmax}{{c_{\max}}}
\newmacro{\homfactor}{a}

\newmacro{\scalar}{F}
\newmacro{\potent}{\phi}
\newmacro{\potentc}{\tilde{\phi}}
\newmacro{\avgutil}{J}
\newmacro{\proba}{\mu}
\newmacro{\estimcost}{f}
\newmacro{\costprivacy}{cp}
\newmacro{\cost}{c}		
\newmacro{\payoffns}{J^{ns}}
\newmacro{\payoffci}{J^{ci}}
\newmacro{\potentns}{\potent^{ns}}
\newmacro{\muj}{\mu_{\mathrm{joint}}}
\newcommand\espj[1]{\mathbb{E}_{\muj}\left[#1\right]} 

\newmacro{\strat}{\lambda}
\newmacro{\strateq}{\strat^*}
\newmacro{\strateqt}{\tilde{\strat}^*}
\newmacro{\stratb}{\boldsymbol{\strat}}
\newmacro{\stratset}{\mathcal{F}}
\newmacro{\stratbeq}{\stratb^*}
\newmacro{\stratbeqt}{\tilde{\stratb}^*}
\newmacro{\stratbt}{\tilde{\stratb}}
\newmacro{\stratt}{\tilde{\strat}}
\newmacro{\lambdans}{\lambda_{\textrm{ns}}}
\newmacro{\lambdabns}{\lambdab_{\textrm{ns}}}
\newmacro{\lns}{\ell_{\textrm{ns}}}
\newmacro{\stratsingle}{\lambda_{\textrm{single}}}

\newmacro{\nonstrat}{M}

\newmacro{\des}{\nu}
\newmacro{\optdes}{\des^*}
\newmacro{\optdest}{\tilde{\des}^*}
\newmacro{\supportsize}{k}

\newmacro{\degree}{{d}}
\newmacro{\abscisse}{x}

\newmacro{\num}{k}

\newmacro{\const}{C}
\newmacro{\constmax}{D}
\newmacro{\constmin}{d}
\newmacro{\exponent}{{\alpha}}
\newmacro{\secexponent}{{\beta}}

\usepackage[normalem]{ulem}

\newcommand{\eps}{\varepsilon}

\title{Asymptotic Degradation of Linear Regression Estimates with Strategic Data Sources}

\author{
  Benjamin Roussillon\thanks{Add info} \\
Univ. Grenoble Alpes, Inria, CNRS, Grenoble INP, LIG\\
  \texttt{benjamin.roussillon@inria.fr} \\
    \And
Nicolas Gast \\
Univ. Grenoble Alpes, Inria, CNRS, Grenoble INP, LIG\\
\texttt{nicolas.gast@inria.fr}
   \And
 Patrick Loiseau\\
  Univ. Grenoble Alpes, Inria, CNRS, Grenoble INP, LIG\\
  \texttt{patrick.loiseau@inria.fr} \\
  \And
 Panayotis Mertikopoulos\\
  Univ. Grenoble Alpes, Inria, CNRS, Grenoble INP, LIG, Criteo AI Lab\\
  \texttt{panayotis.mertikopoulos@inria.fr} 
}

\begin{document}
\maketitle

\begin{abstract}
%
%
We consider the problem of linear regression from strategic data sources with a public good component, i.e., when data is provided by strategic agents who seek to minimize an individual provision cost for increasing their data's precision while benefiting from the model's overall precision. In contrast to previous works, our model tackles the case where there is uncertainty on the attributes characterizing the agents' data---a critical aspect of the problem when the number of agents is large.  We provide a characterization of the game's equilibrium, which reveals an interesting connection with optimal design. Subsequently, we focus on the asymptotic behavior of the covariance of the linear regression parameters estimated via generalized least squares as the number of data sources becomes large. We provide upper and lower bounds for this covariance matrix and we show that, when the agents' provision costs are superlinear, the model's covariance converges to zero but at a slower rate relative to virtually all learning problems with exogenous data. On the other hand, if the agents' provision costs are linear, this covariance fails to converge. This shows that even the basic property of consistency of generalized least squares estimators is compromised when the data sources are strategic.

\end{abstract}



\section{Introduction}
\label{sec:introduction}

Consider a linear regression problem consisting of $n$ data points $(x_i, y_i)_{i\in \{1, \cdots, n\}}$, where $x_i$ is a vector of independent variables and $y_i \in \R$ is the corresponding response variable.
Assuming that these variables are linked by a linear model
\abovedisplayskip=1mm
\begin{equation}
\label{eq:linear}
y_i
	= \betab^{\top} x_i + \eps_i
\end{equation}
\belowdisplayskip=1mm
where the $\eps_i$ are mutually independent noise variables, an analyst aims at estimating the parameter vector $\betab$.
If the variance of each $\eps_{i}$ is known and uniformly bounded in $n$ (but not necessarily identical across $i$), the most widespread algorithm to estimate $\betab$ is the famous generalized least squares (GLS) estimator, which is well-known to enjoy important statistical properties.
In particular, for any fixed $n$, Aitken's theorem \cite{Aitken35a} shows that GLS is \emph{BLUE}, i.e., best among linear unbiased estimators.
GLS is also \emph{consistent} (i.e., it converges in probability towards $\betab$ as $n \!\to\! \infty$), and its covariance matrix decreases to zero at a rate $\Theta(1/n)$\footnote{The notation $g(n)=\Theta(f(n))$ indicates that functions $f$ and $g$ grow at the same rate as $n$ goes to infinity.} \cite{Gyorfi02a,hastie_09_elements-of.statistical-learning}.

In a number of recent applications, however, the assumptions underlying those statistical properties do not hold
because the data is provided by strategic agents who incur a cost for providing high-precision (i.e., low-variance) data.
There can be multiple reasons.
If the data is sensitive personal information (as in medical applications),
revealing it with high precision entails a privacy cost that might incentivize individuals to decrease the disclosure precision
\cite{Warner65a,Duchi13a}.
Also, producing high-precision data may require a certain amount of effort (possibly monetary):
this is the case in crowdsourcing \cite{Dasgupta13a} or recommender systems \cite{Ekstrand11a,Avery97a,Harper05a} where providing content or feedback requires effort, or in applications where the data is produced by costly computations.

Considerations of this kind have become central in
an emerging literature on learning with \emph{strategic data sources}, i.e., when the precision of the provided data incurs a cost for the agent providing it.
This literature mainly examines the design of monetary incentive mechanisms to optimize the model's error assuming that agents maximize their rewards minus their individual provision costs, see e.g.,  \cite{Cai15a,Liu16a,Westenbroek20a} and references therein.
In many applications, however, the underlying model also has a \emph{public good component}---i.e., the agents \emph{also} benefit from the model's precision.
This is the case in recommender systems (where users benefit from the overall service quality),
medical applications (where individuals benefit from the data analysis through improved treatments or better healthcare advice),
federated learning \cite{yang2019federated,konevcny2016federated,geyer2017differentially}, etc.
An additional issue in such applications is that the number of participating agents is typically large, so there is a commensurate degree of uncertainty regarding the state or incentives of other agents.
In this context, the validity of standard results on linear regression are longer guaranteed:
in particular, recent work has shown that the GLS estimator is no longer BLUE in the presence of strategic data sources \cite{Gast20a}.
Going deeper, this leads to the following open questions:
\emph{Does GLS remain consistent in the present of strategic agents?
And, if so, does it still enjoy a $\Theta(1/n)$ convergence rate as in the non-strategic case?}

\paragraph{Our contributions.}

In this paper, we provide answers to these questions by means of a data provision model that accounts for both factors identified above: a public good component and uncertainty regarding the agents' types (their private data).
Specifically, we propose a \emph{linear regression game} in which the $i$-th agent's type is characterized by a $d$-dimensional \emph{attribute vector} $x_i \in \R^{d}$ which is drawn i.i.d. across agents
(but is otherwise assumed to be private information).
This attribute vector is the primitive data associated to each agent and forms the basis of the linear model \eqref{eq:linear}: each agent decides
the precision of the response variable $y_i\in\R$ that is revealed to the analyst as a function of $x_{i}$.
This choice is intended to balance the agent's data provision cost against a public good benefit that depends on the precision of the overall model;
importantly, this choice is also made under uncertainty, because the players' attribute vectors are not assumed a priori known.

In this setting, we obtain the following general results:
\vspace{-1mm}
\begin{enumerate}[1.,leftmargin=1.2em]
\item
We provide an explicit characterization of the game's equilibria for different families of data provision cost functions.
Specifically, if the data provision costs are linear in the precision of the disclosed response, the equilibrium distribution of precisions over the space of attribute vectors corresponds to the solution of an optimal design problem.
By contrast, this characterization is no longer valid if the data provision costs are superlinear. \vspace{-1mm}
\item
Subsequently, we analyze the precision of the estimated model in the limit $n\to\infty$.
In this asymptotic regime, our main result is that, for superlinear costs, the GLS estimator remains consistent, but its covariance decreases to zero at a rate \emph{slower} than the standard $\Theta(1/n)$ rate.
Surprisingly, as the data provision costs become approximately linear, this rate becomes progressively slower, to the point that the GLS estimator \emph{fails to be consistent} if the data provision costs are linear.
\vspace{-1mm}
\end{enumerate}

Our analysis reveals that the key reason behind this asymptotic degradation of the GLS estimator is the following:
as $n\to\infty$, the response provided by each agent at equilibrium tends to deteriorate because of the increasing free-riding effect inherent to public good games.
When the data provision costs are linear, this decrease cannot be compensated by the increase in the number of data points, so the GLS estimator becomes inconsistent.
In this regard, our results illustrate clearly
how free-riding can disrupt even the most fundamental statistical properties of GLS estimators.

\paragraph{Related work.}

There is a growing body of works on scenarios where one wants to learn from data provided by sources that choose their effort when generating data \cite{Cai15a,Luo15a,Liu16a,Westenbroek20a}.
These works assume that the data sources maximize a monetary incentive minus effort exerted and look for mechanisms that minimize the model's error under the assumption that the analyst collecting data cannot see the effort exerted by the data sources.
The data elicitation and crowdsourcing literatures contain similar mechanism design problems for cases where either the effort exerted or the data report (or both) are unverifiable \cite{Frongillo15a,Dasgupta13a,shah2016double, Kong_Schoenebeck_Tao_Yu_2020}.
More broadly, there is an important literature on mechanism design for statistical estimation problems that assumes that the data sources are strategic in some way, notably for cases where agents may lie on their cost for revealing data \cite{Abernethy15a,Chen18b,Chen19a} (see also related problems of mechanism design in the context of differential privacy \cite{ghosh-roth:privacy-auction,Dwork14a}). \looseness=-1

Several works analyze mechanism design problems related to linear regression with strategic data sources, where the agents directly report their response variable $y_i$ (or their input variable $x_i$) and may lie about it or strategically optimize it \cite{perote2004,dekel2010incentive,Caragiannis16a,Chen18a,Ben-Porat19a,Hossain20a,Shavit20a, chen2020truthful, golrezaei2021dynamic, amin2014repeated, amin2013learning} (see also similar problems in the context of classification \cite{Meir12a,Hardt16a,Dong18a,Milli19a,Kleinberg19a,Zhang19a,Miller20a,Tsirtsis20a,Braverman20a}).
In contrast, we assume that the agents choose the precision of the data provided.
More importantly, the fundamental difference is that those works all assume that the agents are motivated by the accuracy or decision of the learned model in their own direction while we assume that agents equally benefit from the public good component.

Games with a public good component have been the subject of many studies in economics (see \cite{Morgan00a} and references therein);
however, in all these works the public good is simply the sum of contributions from all agents.
To the best of our knowledge, the only work that considers a public good component in the context of learning from strategic data sources is \cite{Gast20a} (see also an earlier version of the same in \cite{Ioannidis13a} and \cite{Chessa15c} in the simple case of estimating a population average), which is perhaps the closest antecedent to our paper. 
%
The authors of \cite{Gast20a} propose a game-theoretic model with a public good component and \emph{common knowledge} of player types (i.e., the agents' attribute vectors $x_i$ are public). We build on their model, but we introduce uncertainty on the agents' attributes (i.e., they are considered private)---which is crucial to frame our main questions (link to optimal design and asymptotic precision) rigorously. More importantly though, our analysis is of a different nature than that of \cite{Gast20a}, both technically and conceptually. The analysis of \cite{Gast20a} only concerns the game's \emph{price of stability} (PoS) and the validity of Aitken's theorem.
Specifically, \cite{Gast20a} shows that, for any given $n$, GLS may fail to be BLUE (i.e., Aitken's theorem fails), but it is otherwise ``uniformly close'' to optimal: the improvement ratio relative to any other linear unbiased estimator is independent of $n$, so GLS is still ``order optimal'' in their model.
In contrast, our results show that GLS (and other linear unbiased estimator by proxy) does not even produce the same convergence \emph{rate} as non-strategic regression. This means that characterizing this convergence rate is ultimately more important to the analyst than choosing an ad-hoc linear unbiased estimator which can only lead to constant-term multiplicative improvement. To ease the comparison of our work to \cite{Gast20a}, we include in Section~\ref{sec:preliminary} an adaptation of their results for our model with uncertainty, we then provide a detailed technical comparison in further sections that contain our new results. Note finally that our model with uncertainty enables uncovering an interesting connection to optimal design, which \cite{Gast20a} does not touch upon. 
We also note that our game has the structure of an aggregative game in the sense of \cite{cornes2012fully} which, however, does not offer any further insights.

\section{Problem setup and assumptions}
\label{sec:setup}


\paragraph{The linear regression game}

We consider a model of strategic data sharing in which a group of $n$ agents wants to collectively learn a linear model $y \approx \betab^{\top} x$.
Here, $x$ is a $d$-dimensional vector, $y$ is a scalar and the vector $\betab$ represents the weights of the linear model that agents want to estimate. Each agent $i\in N:=\{1, \cdots, n\}$ is associated to an \emph{attribute vector} $x_i$ which is drawn i.i.d. across agents from a set of possible attribute profiles $\Xset \subseteq \R^{d}$.
Then, based on these attributes, each agent can select
a \emph{precision level} $\lambda_i(x_i)\in\R_+$ and produce an unbiased estimate $\hat{y}_i$ of $\betab^{\top}x_i$ with this precision.%
\footnote{We can impose an upper bound $\lambdamax$ on the precision that an agent can choose; our results would still hold for large-enough $n$ as the constraint is never binding.
In the sequel, we assume $\lambdamax = \infty$ to simplify the exposition.\looseness=-1}
More explicitly, the response variable reported by the $i$-th agent is
\abovedisplayskip=1mm
\begin{equation}
\hat{y}_i
	= \betab^{\top}x_i+\eps_i,
\end{equation}
\belowdisplayskip=1mm
where $\eps_i$ is an error term of mean $0$ and variance $1/\lambda_i(x_i)$.
%
Agents send this estimate $\hat{y}_i$, along with their values of $x_i$ and of the precision $\lambda_i(x_i)$ to an aggregator that publicly discloses an estimate $\betabh$.
The errors terms $\eps_i$ are assumed to be independent, but we do not make any further assumption on their distribution.
We then posit that each agent tries to balance a trade-off between two components:
\vspace{-1mm}
\begin{enumerate}[leftmargin=1em]
\item
\emph{Idiosyncratic cost:}
The value $\hat{y}_i$ may be either sensitive or costly to produce.
It is sensitive for example when it represents a disease likelihood, a total debt or any attribute that might hurt the agent if it is disclosed with full precision (e.g., by a potential increase in cost of health insurance):
here, the agent possesses a value $y_i$ but only reveals a noisy version of it $\hat{y}_i$.
It is costly to produce when it is the result of a simulation involving heavy computations, or when it requires human work.
We represent all these scenarios by assuming that releasing an estimator $\hat{y}_i$ with precision $\lambda_i(x_i)$ induces a cost $c_i(\lambda_i(x_i))$ to agent $i$. We refer to it as the \emph{\textbf{(data) provision cost}}.

\item
\emph{Public good benefit:}
A key feature of our model is that all agents benefit from the learned model $\betabh$.
For example, in a medical context, agents would be interested to know that a given disease is correlated to their weight or cholesterol level;
in recommender systems, agents might be interested to know what affects the good rating of a restaurant;
etc.
We model this benefit as a \emph{public good}, that is, we assume that each agent benefits equally from the estimated model's precision---which, in turn, depends on each agent's prescribed precision.
As it is easier to maintain a cost-oriented perspective, we represent this by considering that each agents incurs a cost $\Cpublic{\lambdab}$, where $\lambdab = [\lambda_{\play}]_{i\in N}$.
We refer to it as the \emph{\textbf{estimation cost}}. \blue{It should be noted that this public good aspect is central to the situations we model and to the results we obtain. Indeed, our proofs rely on the fact that the game is a potential game (as stated in Proposition~\ref{thm.uniqueness}). This property is implied by the public good aspect.}
\end{enumerate}

\begin{remark*}
This model is particularly relevant in the context of federated learning \cite{yang2019federated}.
There, each agent performs a local estimation and the estimations are combined to get a model.
This paradigm can be used for reasons of efficiency (many agents, perform a local optimization \cite{konevcny2016federated}) or privacy (agents want to compute a joint representative model without explicitly having to share their personal data \cite{geyer2017differentially}).
Our model can be viewed as an instance of both cases. \looseness=-1
\end{remark*}

To proceed, we model the collective behavior of agents by considering a game in which each agent $i\in N$ chooses their strategy $\lambda_i:\Xset\to \R_+$ to minimize their cost $J_i(\lambda_i,\lambdab_{-i})$, defined here as
\begin{align}
\label{eq:individual_cost}
J_i(\lambda_i,\lambdab_{-i})
	= \esp{c_i(\lambda_i(x_i))} + \Cpublic{\lambdab},
\end{align}
where the expectation is taken
with respect to the law $\mu$ of the attribute vectors $x_{i}$.
Here, $\lambdab_{-i}$ denotes the collection of strategies of all agents except the $i$-th one,
and $\lambdab = (\lambda_i, \lambdab_{-i})$ will be called a \emph{precision profile}.
Note that, given that agent $i$ chooses the function $\lambda_i:\Xset\to \R_+$, minimizing the expected provision cost $\esp{c_i(\lambda_i(x_i))}$ as in \eqref{eq:individual_cost} is equivalent to minimizing the cost for each value of $x_i$ separately. 
On the other hand, the definition of $\Cestim$ is given below and involves an expectation on all agents' attributes, which models the uncertainty of an agent about other agents' attributes.

The setting described above defines a game that we refer to as the \emph{linear regression game}.
We emphasize that the strategy of each agent is a function $\lambda_i:\Xset\to \R_+$, i.e., each player's strategy space is the $|\Xset|$-dimensional orthant $\R_{+}^{\Xset}$.
Throughout the paper, to avoid confusion, we will denote such functions with the greek letter $\lambda$ and we will use the latin letter $\ell$ for scalar values such as $\lambda_i(x_i)$.
In the sequel, we analyze the Nash equilibrium of this linear regression game.
A precision profile $\lambdabeq$ is a Nash equilibrium of the game if for all $i \in N$, $\lambdaeq_i$ minimizes $J_i(\lambda_i,\lambdabeq_{-i})$.

\paragraph{Generalized least squares, definition of $\Cestim$}

The analyst receives the $\nbplay$ triplets $(x_i,\hat{y}_i,\lambda_i(x_i))$ and uses them to produce an estimate $\betabh$ that is then sent to the agents.
We assume that the analyst computes this estimate using \emph{generalized least squares} (GLS) and denote it $\betabh_\gls$.
GLS is a generalization of the least squares regression to heteroscedastic data, that is, when the precisions $\lambda_i(x_i)$ of the different data points are different.
It is one of the most widespread estimators for this scenario, in particular because by Aitken's theorem, GLS is optimal in that, for given precisions, it has the smallest covariance (in the positive semi definite sense) among all linear unbiased estimators \cite{Aitken35a}.
The covariance of GLS is independent of $\hat{y}_i$ and equal to $\p{\sum_{i} \lambda_i(x_i) x_ix_i^{\top}}^{-1}$.
Note that this quantity is well defined only if the matrix $\sum_i \lambda_i(x_i) x_ix_i^{\top}$, called the information matrix, is invertible; if not, the estimator $\betabh_\gls$ is not unique as the generalized least squares problem has infinitely many solutions.

Recall that the values $x_i$ are generated
randomly according to a common distribution $\mu$ on $\Xset$.
In what follows, we
consider that the estimation cost is a function of
the expected information matrix, that is:\looseness=-1
\begin{align}
  \label{eq:Upublic}
  \Cpublic{\lambdab} &= \scalar\Bigg(\p{\esp{\sum_{i\in N} \lambda_i(x_i) x_ix_i^{\top}}}^{-1}\Bigg),
\end{align}
where $\scalar: S_+^d \rightarrow \R_+$ is a so-called \emph{scalarization function} that maps the covariance of the estimator to a
cost (we denote by $S^d_+$ the set of positive semidefinite
matrices of size $d\times d$).
Scalarizations are standard in optimal design (see Section~\ref{sec:equilibrium}), and they include standard metrics of a model's quality (such as the mean squared error) as special cases---see details in Appendix~\ref{sec:scalarizations}. 

The public good component \eqref{eq:Upublic} is a function of the inverse of the expected information matrix.
In particular, agent $i$ is included in this expectation, so they minimize a function that includes their individual contribution $x_i$. In this regard, \eqref{eq:Upublic} can be seen as a ``middle ground'' between the approach of \cite{Gast20a} (which assumes complete information of the $x_i$ of each agent), and a Bayesian model where agents would optimize over $\esp{F((\sum_{i\in N} \lambda_i(x_i) x_ix_i^{\top})^{-1})}$. The former is impractical in asymptotic settings while the latter introduces a series of modeling artifacts due to the nonzero probability of encountering an ill-defined linear regression problem when drawing vectors from a finite set.\looseness=-1

Compared to the above, our model requires the same information as the Bayesian framework, but it does not face the same issues. 
In addition, it is possible to establish a precise equivalence between our game and the complete information game when the number of players is large---see Appendix~\ref{sec:equiv} for the details.
Note also that our model relies on GLS, which requires the analyst to know the precision of each data point (i.e., to know the heteroscedastic noise levels). We stress here, however, that our main result also applies to the \emph{ordinary least squares (OLS)} estimator, which does not require this information. We discuss this in detail in Section~\ref{sec:asymptotic}. Note that our results can also be extended to the wider class of regression procedures whose covariance satisfy the convexity and homogeneity assumptions detailed below as these are the only properties of GLS and OLS we use in our proofs.
\looseness=-1

\paragraph{Technical assumptions}

Through our analysis, we make the following assumptions:
\begin{assumption}
  \label{structure.information}
  The set $\Xset$ is finite, $\mu$ has full support on $\Xset$, and $\esp{x_ix_i^{\top}}$ is positive definite.
\end{assumption}
\begin{assumption}
  \label{scalarization.type}
  The scalarization $\scalar: S_{+}^\ndim \rightarrow \R_+$ is
  non-negative, increasing in the positive semidefinite order, and
  convex. $\scalar$ is homogeneous of degree $q$: 
  $\forall a>0, \forall V\in S_+^\ndim$, $\scalar(a\cov) = a^{\qhomo} \scalar(\cov)$.
\end{assumption}
\begin{assumption}
  \label{privacy.type}
  The provision costs $\cost_\play: \R_+ \rightarrow \R_+$ are
  non-negative, increasing, and convex.
\end{assumption} 

Assumption~\ref{structure.information} is a technical assumption that guarantees that the game is well defined and non-trivial.
Specifically, the condition $\esp{x_ix_i^{\top}} \succ 0$ simply implies that the information matrix is invertible for some $\lambdab$, while the finiteness of $\Xset$ simplifies the mathematical formalism and proofs by avoiding subtle compactness issues in the definition of an equilibrium.
Assumption \ref{scalarization.type} ensures that the estimation cost is non-decreasing and convex, which is a very standard setting in game theory.
Together with the homogeneity assumption, these conditions are satisfied by all the commonly used scalarizations in statistics and optimal design such as the trace ($\qhomo \!=\! 1$, A-optimal design), squared Frobenius norm ($\qhomo \!=\! 2$), or mean squared error ($\qhomo \!=\! 1$, I/V-optimal design).

Assumption \ref{privacy.type} is also common. The convexity of the provision costs implies that there is a non-decreasing marginal cost to increase the precision of the data provided. This is reasonable and includes both the linear and superlinear cases, each being relevant. For instance, if data points represent an average over multiple measurements, the precision depends linearly on the number of measurements. If the precision depends on simulations (e.g., involving a discrete search of a continous space), 
obtaining a higher precision might require a polynomial increase in computation time.
\looseness=-1




\section{Preliminaries and First Results}
\label{sec:preliminary}
%

In this section, we discuss structural results for the linear regression game.
In the spirit of \cite{Gast20a},
we show that our game is a potential game and provide a bound on its price of anarchy.\looseness=-1




For a given precision profile, we define $\potent(\lambda_i,\lambdab_{-i})$ as
\begin{equation}
  \label{avg.potential}
  \potent(\stratb) = \sum_{j = 1}^\nbplay
  \esp{\cost_j(\strat_j(x))} + \Cpublic{\stratb}.
\end{equation}
Using the form of $J_i(\lambda_{i},\lambdab_{-i})$ in \eqref{eq:individual_cost}, a strategy $\lambda_i$ minimizes $J_i(\lambda_i,\lambdab_{-i})$ over all possible strategies $\lambda_i$ (for a fixed $\lambdab_{-i}$) if and only if it minimizes \eqref{avg.potential}.\footnote{It is easy to verify that for every $\lambda_{i}$, $\lambdab_{-i}$ and $\lambda^{\prime}_{i}$, we have $J_i(\lambda_{i},\lambdab_{-i}) - J_i(\lambda^{\prime}_{i},\lambdab_{-i}) = \potent(\lambda_{i},\lambdab_{-i}) - \potent(\lambda^{\prime}_{i},\lambdab_{-i})$.}
Since function $\phi$ is independent of $i$, this shows that the game is a potential game \cite{Neyman97a} and $\potent$ is the potential of the game.
As stated in the next proposition (whose proof is in Appendix~\ref{sec:proofs}), expressing the game as a potential game simplifies the study of its Nash equilibria by transforming it into the easier problem of studying the minima of a convex function.\looseness=-1
\begin{proposition}
  \label{thm.uniqueness}
  Under Assumptions~\ref{structure.information}, \ref{scalarization.type}, and \ref{privacy.type}, a precision profile $\lambdabeq$ is a Nash equilibrium of the linear regression game if and only if it minimizes $\potent$.
 Such an equilibrium exists.
It is unique if all provision cost functions $c_i$ are strictly convex.
 When there are multiple equilibria, the estimation cost $\Cpublic{\lambdabeq}$ does not depend on the equilibrium.
\end{proposition} 
%




The price of anarchy (PoA) is a standard concept in game theory that characterizes the degradation of performance due to players' selfish behavior.
It is the ratio between the total cost of the worst Nash equilibrium and the minimal achievable total cost. In our setting the total cost is $\SC(\lambdab) = \sum_{i} \esp{c_i(\lambda_i(x))}+n\Cpublic{\lambdab}$. Hence, denoting as $\NE\subseteq \{\lambda:\Xset\to \R_+\}^n$ the set of Nash equilibria,  
\begin{align*}
    \poa = \frac{\max_{\lambdabeq\in\NE} \SC(\lambdabeq)}{\min_{\lambdab\in \{\lambda:\Xset\to \R_+\}^n} \SC(\lambdab)}.
\end{align*}


Our linear regression game has the same PoA bound as that of \cite{Gast20a} with a similar proof (see App.~\ref{sec:proofs}):
\begin{theorem}\label{generalizedBound}
  In addition to Assumptions \ref{structure.information}, \ref{scalarization.type} and \ref{privacy.type}, assume that there exist $\pmin\ge1$ such that for all $i\in N, a>1, \ell>0$: $c_i(a\ell) \ge a^{\pmin}c_i(\ell)$.
Then, the price of anarchy satisfies $\poa \leq {n}^{\frac{q}{\pmin + q}}$.
 Additionally, for all $\varepsilon > 0$, there exists a game such that $\poa \geq {n}^{\frac{q}{\pmin+q}}(1-\varepsilon)$.
\end{theorem} 
%
%
\begin{remark*}
We should note here that the above result bounds the game's price of anarchy whereas the study of \cite{Gast20a} concerns the price of stability ($\pos$) of a suitable variant of our game without uncertainty.
In contrast to the price of anarchy, the price of stability compares the social optimum cost to that of the \emph{best} Nash equilibrium.
In general, these two measures of selfishness can vary wildly, but in the linear regression game under study, they conincide;
this is due to the fact that although we may have multiple equilibria, all equilibria have the same cost (from Proposition~\ref{thm.uniqueness}). This differs from \cite{Gast20a} where there exists a unique non-trivial equilibrium, but there also exist trivial equilibria with infinite costs.\looseness=-1
\end{remark*}


\section{Characterization of the equilibrium}
\label{sec:equilibrium}


We now characterize how the attribute distribution $\mu$ affects the precision given by agents at equilibrium. We show that when provision costs are linear, the precision given by each agent can be mapped to the solution of an optimal design. This is no longer true when provision costs are not linear.



In optimal design \cite{pukelsheim2006optimal,atkinson2007optimum,boyd2004convex}, an analyst chooses the $x_i$'s of the set of (non-strategic) data sources in order to maximize the quality of the linear model estimated via a scalarization of the covariance matrix. Formally, the optimal design problem for the scalarization $\scalar$ (see Appendix~\ref{sec:scalarizations} for details) and the design space $\Xset$ is to find a probability measure $\optdes$ that minimizes:\looseness=-1
\begin{equation}
 \label{eq.optdes}
 \optdes \in \argmin_{\des} \scalar\left(\left(\sum_{x\in\Xset} x x^{\top} \des(x)\right)^{-1}\right).
\end{equation}

%

In our linear regression game, the agents have an incentive to produce a useful information matrix to minimize the estimation cost but they are limited by the inherent allocation $\mu$ of attribute vectors and by the provision costs $c_i$.
An equilibrium is a minimum of the potential \eqref{avg.potential} that contains the estimation cost $\Cpublic{\lambdab}$, which can be rewritten as:
\begin{align}
 \Cpublic{\lambdab}&=\scalar\p{\p{\sum_{x\in\Xset}
 xx^{\top}\sum_{i\in N}\lambda_i(x)\mu(x)}^{-1}}.
 \label{eq:Upublic_integral}
\end{align}

The similarity of \eqref{eq.optdes} and \eqref{eq:Upublic_integral} suggests a link between the Nash equilibria of the linear regression game and the solutions of the optimal design problem on $\Xset$ by interpreting $\sum_i\strat_i(x)\mu(x)$ as a design $\des(x)$:\looseness=-1
%
\begin{theorem}
 \label{thm.optim.design}
 Consider a linear regression game that satisfies Assumptions~\ref{structure.information} and \ref{scalarization.type} and such that all provision costs are linear (\emph{i.e.}, $c_i(\ell)=a_i\ell$ for all $i\in N$ and $\ell \in \R_{+}$, where $a_i$ is a constant). Let $\lambdabeq$ be a Nash equilibrium and let $\nu_{\lambdabeq}$ be the measure such that $\nu_{\lambdabeq}(x)=\sum_{i\in N}\lambda^*_i(x)\mu(x)$ for all $x\in \Xset$. Then, the probability measure defined by $\nu_{\lambdabeq}(x)/\sum_{y\in\Xset}\nu_{\lambdabeq}(y)$ is an optimal design of \eqref{eq.optdes}.
\end{theorem} 
\begin{proof}[Sketch of proof]
 A detailed proof is given in Appendix~\ref{sec:proofs}. The main idea is to see the minimization problem \eqref{eq.optdes} as an optimization problem with constraint $\sum_{x\in\Xset}\nu(x)=1$. When the provision costs are linear, the potential $\phi$ is a Lagrangian of this optimization problem with a dual variable $\min_{i\in N}a_i$. The fact that $\nu_{\lambdabeq}$ is proportional to an optimal design is then a consequence of the homogeneity of the scalarization (Assumption~\ref{scalarization.type}).
\end{proof} \vspace{-3mm}
%
While the shape of $\nu_{\lambdabeq}$ for an equilibrium $\lambdabeq$ is that of an optimal design, the total expected precision $\sum_{x\in\Xset}\nu_{\lambdabeq}(x)$ depends on the provision costs. Theorem~\ref{thm.optim.design} merely states that agents contribute proportionally to an optimal design but does not characterize how the total precision depends on the number of agents or on the agents' costs. We leave this discussion to Section~\ref{sec:asymptotic} (in particular Theorem~\ref{thm.conv}). This theorem also implies that with linear costs, agents that have data points which do not belong to an optimal design are pure free-riders. On the contrary, this is no longer the case with superlinear provision costs.  We illustrate this in Figure~\ref{fig.precs} where we observe that the difference between the maximum and minimum precision given depending on the data point shrinks as the exponent of the cost grows.


The particular connection between optimal design and Nash equilibria exhibited in Theorem~\ref{thm.optim.design} is tightly connected to the linearity of provision costs. When costs are strictly convex, the allocation of precision across $\Xset$ at equilibrium is in general suboptimal. For instance, if an agent has a provision cost $c_i(\ell)=\ell^p$ with $p>1$, then the derivative of this provision costs at $0$ is zero, $c_i'(0)=0$. In such a case, this agent will provide a positive precision, $\lambda_i(x)>0$, for all attribute vectors $x\in\Xset$ even though the support of an optimal design might be smaller than $\Xset$. 
We illustrate the case of nonlinear costs in a polynomial regression setting that is an instance of our linear regression game as follows.
Let $\Xset = [1, \abscisse, \cdots, \abscisse^{\degree-1}]^{\top}$ be the set of attribute vectors with $x\in\{-10\dots 10\}$.
We compare in Figure~\ref{fig.precs} the allocation of precision at equilibrium $\nu_{\lambdabeq}$ as defined in Theorem~\ref{thm.optim.design} to the optimal design for different monomial provision costs ($\cost(\ell) = \ell^\phomo$). We set $\mu$ to the uniform distribution on $\Xset$, $d=4$, $n=10$ and the scalarization $F$ is the trace. Other parameters give similar results (see App.~\ref{sec:more_experiments}). We observe that when the provision costs are near-linear ($p\!=\!1.01$), the precision function is similar to the optimal design yet different. When $p\!=\!1.2$ or $p\! =\! 3$, however, the precision for the vector $[1,0,\dots,0]$ is maximal whereas the optimal design sets a weight $0$ to it. Intuitively, the convexity of provision costs yields a more spread-out allocation of precision than the optimal design. This shows that equilibrium can be different from optimal design, even when costs are close to linear.

\begin{figure}
 \centering
 \begin{subfigure}{.245\textwidth}
 \includegraphics[width=\textwidth]{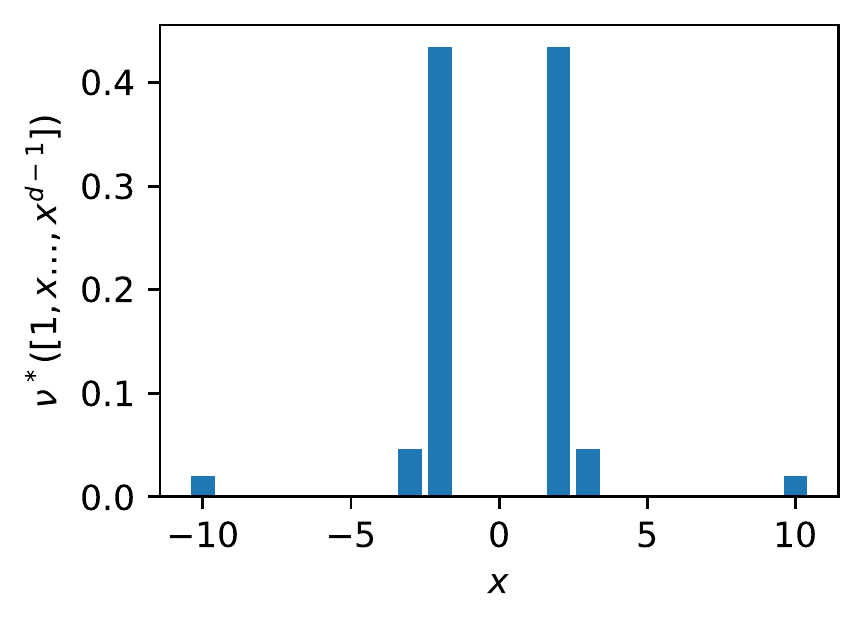}%
\vspace{-3mm}
 \caption{Optimal design $\nu^*$}
 \label{fig.unif.a}
 \end{subfigure}
 \begin{subfigure}{.245\textwidth}
 \includegraphics[width=\textwidth]{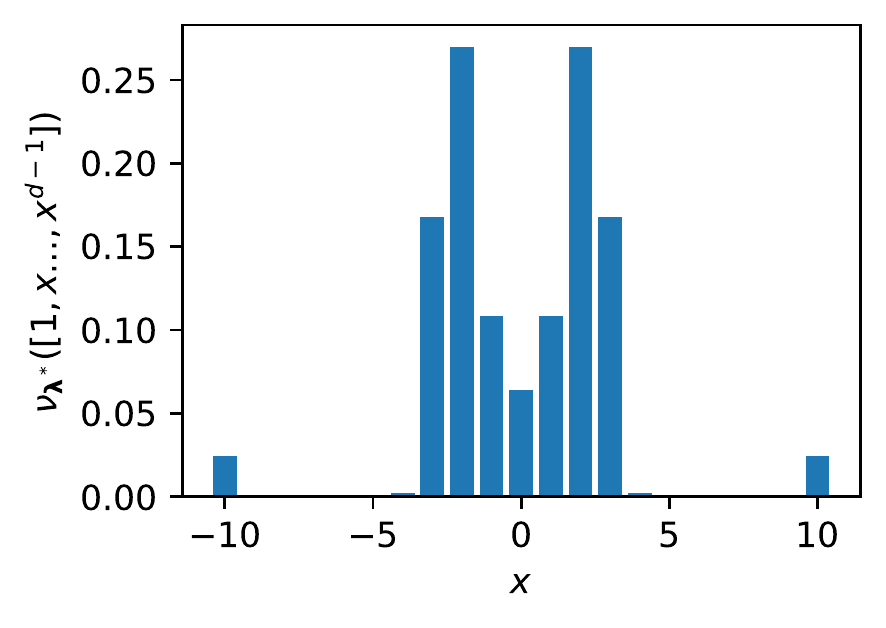}%
\vspace{-3mm}
 \caption{$\nu_{\lambdabeq}$ when $c_i(\ell)\!=\!\ell^{1.01}$}
 \label{fig.unif.b}
 \end{subfigure}
 \begin{subfigure}{.245\textwidth}
 \includegraphics[width=\textwidth]{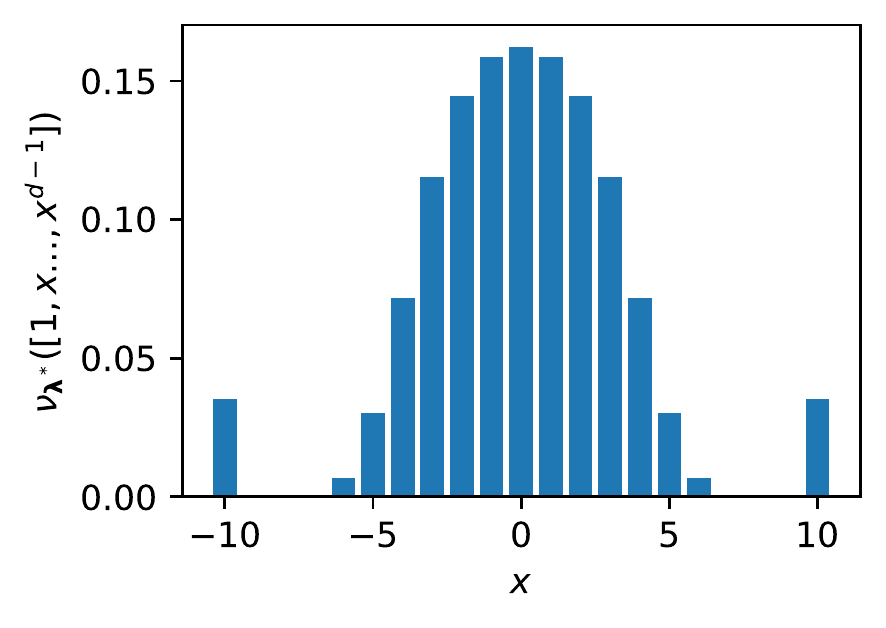}%
\vspace{-3mm}
 \caption{$\nu_{\lambdabeq}$ when $c_i(\ell)=\ell^{1.2}$}
 \label{fig.unif.c}
 \end{subfigure}
 \begin{subfigure}{.245\textwidth}
 \includegraphics[width=\textwidth]{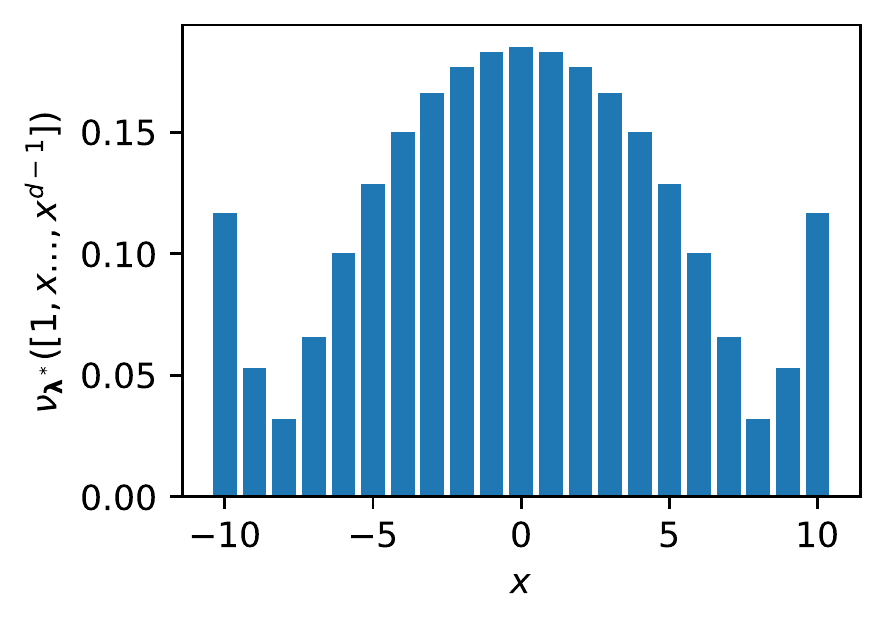}%
\vspace{-3mm}
 \caption{$\nu_{\lambdabeq}$ when $c_i(\ell)=\ell^{3}$}
 \label{fig.unif.d}
 \end{subfigure}
 
 \caption{Optimal design $\nu^*$ and allocation of precision at equilibrium $\nu_{\lambdabeq}$.\vspace{-4mm}}
 \label{fig.precs}
\end{figure}
\vspace{-4mm}



\section{Asymptotic results}
\label{sec:asymptotic}
The previous section shows that linear provision costs drive agents to allocate their precision proportionally to an optimal design, while non-linear costs lead to a non-optimal allocation. In this section, we show that the situation is
\emph{radically different when considering the total model precision}. 

\paragraph{The case of identical agents}

To gain intuition, 
we first consider agents with identical monomial costs.
In this setting, the equilibrium for the $n$-agent game is obtained by scaling the solution of the optimization problem that would correspond to a single-agent game (the full proof is in Appendix~\ref{sec:proofs}): \looseness=-1
\begin{proposition}
  \label{thm.scaling}
  Consider a linear regression game satisfying Assumptions~\ref{structure.information} and \ref{scalarization.type} and such that for all agent $\play$ and precision $\ell$: $\cost_\play(\ell) = \ell^p$ with $p\ge1$. Let $\stratsingle = \argmin_{\strat \in \R_+^{|\Xset|}}\esp{\strat(x)^p} + \Cpublic{\lambda}$. \vspace{-.6cm}
  \begin{itemize}
  \item[(i)] The precision profile $\stratbeq$ with
    $\lambdaeq_i=n^{-\frac{1+q}{p+q}}\stratsingle$ for all $i=1,\dotsc,n$ is a Nash equilibrium. \vspace{-.3cm}
  \item[(ii)] The estimation cost at equilibrium is
    $\Cpublic{\lambdabeq}=n^{-q\frac{p-1}{p+q}}\Cpublic{\stratsingle}$.
  \end{itemize}\vspace{-.3cm}
\end{proposition}


Proposition~\ref{thm.scaling}$(i)$ illustrates a major difference between the strategic and non-strategic settings. Indeed, in a non-strategic setting, each agent would provide data with a fixed precision, say $\lambdans(x) = \lns$ for all $x$. By contrast, in the presence of strategic data sources,
\emph{the equilibrium precision given by each agent goes to $0$ when the number of agents grows.}
Moreover, the convergence rate is governed by the parameters $p$ and $\qhomo$:
when $p\to\infty$, the precision of each agent is almost constant, similar to the non-strategic case;
instead, with linear costs ($p = 1$), the precision given by each agent goes to $0$ at a $\Theta(1/n)$ rate.\looseness=-1

Thus, when aggregating the data from $n$ non-strategic data sources, the estimation cost would be
\begin{equation}
\Cpublic{\lambdabns}
	= n^{-q}\Cpublic{\lambdans}
\end{equation}
where $\lambdabns = (\lambdans, \cdots, \lambdans)$ (which corresponds to the standard $1/n$ rate if $q=1$). By contrast, when aggregating the data from $n$ strategic data sources, Proposition~\ref{thm.scaling}$(ii)$ shows that the rate of decrease is smaller, again governed by the parameters $p$ and $q$. In the extreme, when the costs are linear ($p = 1$),
\emph{the estimation cost does not even go to $0$ as $n\to\infty$.}
This shows that GLS is not consistent in the presence of strategic data sources with linear provision costs:
in this case, the estimator's covariance does not vanish as the number of data sources grows large.

To quantify how strategic considerations lead to a degradation of the GLS estimator,
we can consider the ratio between the strategic and non-strategic estimation
costs:
\begin{equation}
\Cpublic{\stratbeq}/\Cpublic{\lambdabns}
	= \Theta(\nbplay^{\frac{\qhomo(\qhomo + 1)}{p + \qhomo}}).
\end{equation}
This ratio goes to infinity for any possible value of the parameters, implying that strategic agents
\emph{always end up incurring an asymptotic degradation of the GLS estimator as $n\to\infty$.}
In particular, higher values of $\qhomo$ imply a more drastic degradation because the estimation cost is reduced in a neighborhood of $0$, which makes agents less willing to exert effort.
A high $\phomo$ implies a smaller degradation as agents are less sensitive to their provision costs as long as their precision is smaller than $1$. \looseness=-1

Figure~\ref{fig.conv} illustrates the convergence of the estimation cost and the degradation ratio for various values of $p$ and $q$. Figure~\ref{fig.p.vary} pictures the convergence of the estimation cost (in $\nbplay^{-\qhomo\frac{p - 1}{p + \qhomo}}$). It illustrates the inconsistency of GLS when provision costs are linear ($p=1$) and the better convergence rate with larger $p$ and $\qhomo$. In more detail, Figure~\ref{fig.q.vary} depicts the degradation of the estimation cost due to the presence of strategic agents. We observe that the relative position of the curves is different than in Figure~\ref{fig.p.vary}: the degradation ratio is higher for $(p = 2, \qhomo = 3)$ than for $(p = 1, \qhomo = 2)$, whereas the first case yields a consistent estimator and the second does not. This illustrates the dual impact of $\qhomo$ on the linear regression game: a lower $q$ implies a lower estimation cost but also implies a lower effort, making the estimation cost prohibitively high relative to the non-strategic setting.

\begin{figure}[ht]
\vspace{-.3cm}
  \centering
  \begin{subfigure}{0.37\textwidth}
    \centering
    \includegraphics[width=\textwidth]{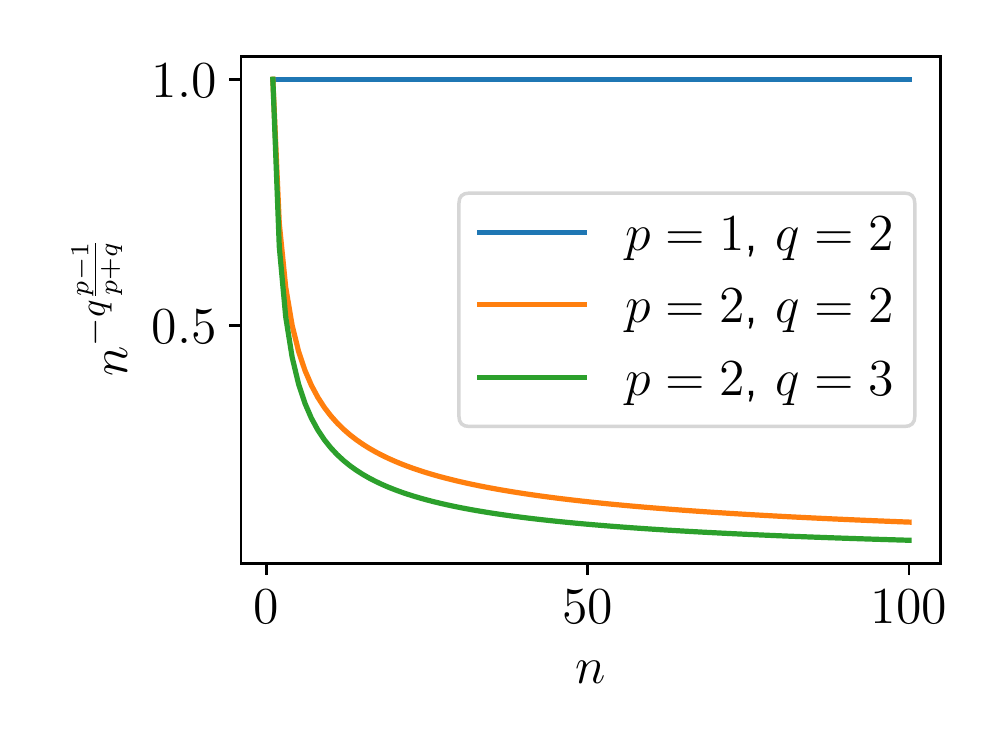}
\vspace{-8mm}
    \caption{Estimation cost $\Cpublic{\stratbeq}$}
    \label{fig.p.vary}
  \end{subfigure}
  \begin{subfigure}{0.37\textwidth}
    \centering
    \includegraphics[width=\textwidth]{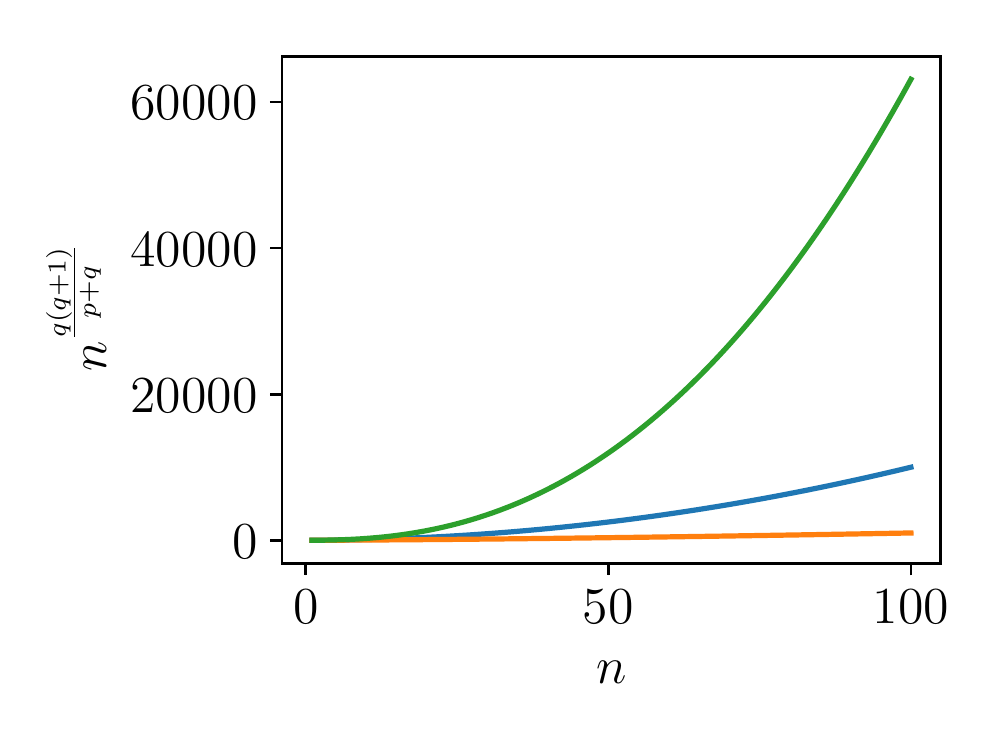}
\vspace{-8mm}
    \caption{\mbox{Degradation ratio $\Cpublic{\stratbeq}/\Cpublic{\lambdabns}$}}
    \label{fig.q.vary}
  \end{subfigure}
  \vspace{-2mm}
  \caption{Influence of $p$ and $q$ on (a) the estimation cost and (b) the degradation ratio.\vspace{-4mm}}
  \label{fig.conv}
\end{figure}

\paragraph{Main result: Asymptotic degradation of estimation cost in the general case} We are now ready to state the main result of the paper, which characterizes the asymptotic behavior of the estimation cost under non-identical and general provision costs. The next theorem provides upper and lower bounds on how the estimation cost decreases as $n\to\infty$. 
\begin{theorem}
  \label{thm.conv}
  Assume that Assumptions~\ref{structure.information}, \ref{scalarization.type} and \ref{privacy.type} hold. Additionally, assume that there exist $\pmin,\pmax\ge1$ and functions $c_{\min},c_{\max}:\R_+\to\R_+$ such that for all $i\in N$ and all $a>1, \ell>0$: $a^{\pmin}c_i(\ell) \le c_i(a\ell) \le a^{\pmax}c_i(\ell)$ and $0<c_{\min}(\ell)\le c_i(\ell)\le c_{\max}(\ell)<\infty$. Then there exist constants $\constmin, \constmax > 0$ that do not depend on $\nbplay$ and such that:\vspace{-1mm}
  \begin{equation}
    \label{eq.conv}
  \constmin \nbplay^{-\qhomo\frac{\pmin - 1}{\pmin + \qhomo}-\alpha} \le \Cpublic{\lambdabeq} \le  \constmax\nbplay^{-\qhomo\frac{\pmin - 1}{\pmin + \qhomo}}, \quad \textrm{ where } \quad \alpha=\qhomo\frac{(\pmax-\pmin)(q+1)}{\pmax(q+\pmin)}.
\end{equation}
\end{theorem}
\vspace{-4mm}

\begin{proof}[Sketch of proof]
A full proof is given in Appendix~\ref{sec:proofs}.
To get the upper bound, we first obtain an upper bound of the potential $\phi$ by evaluating it on a well-chosen precision profile $\lambdab$ inspired by Proposition~\ref{thm.scaling}. Combining this with the assumption that $a^{\pmin}c_i(\ell) \le c_i(a\ell), \forall \ell\in\R_+$ and with the homogeneity of the estimation cost gives the right-hand-side of \eqref{eq.conv}. 
The lower bound is harder to get. We first exploit the previous upper bound to get an upper bound on the total provision cost (the left part of the potential \eqref{avg.potential}). Using the assumption that  $c_i(a\ell) \le a^{\pmax}c_i(\ell), \forall \ell\in\R_+$, we deduce an upper bound on the total precision. We then consider an optimal design scaled with this total precision and show, using the estimation cost homogeneity, that it gives the left-hand-side of \eqref{eq.conv}. 
From this sketch of proof, observe that the constant $d$ involves the estimation cost of an optimal design while the constant $D$ involves the estimation cost of a non-strategic precision profile $\Cpublic{\lambdabns}$.
  %
\end{proof} \vspace{-.3cm}

Theorem~\ref{thm.conv} is our main result: it characterizes the decay of the GLS estimates covariance with strategic data sources for general data provision costs that satisfy a mild assumption governed by the two parameters $\pmin, \pmax$. This assumption roughly specifies that the provision costs grow faster than $\ell^{\pmin}$ and slower than $\ell^{\pmax}$; it is satisfied for instance by a sum of monomial terms with exponents between $\pmin$ and $\pmax$ and such that coefficients do not vanish or explode. Note that the result of Theorem~\ref{thm.conv} trivially implies that the precision of each agent goes to $0$ when the number of agents grow. Although we do not formally prove it, it is clear that the result remains valid even when the assumptions on costs are valid only in a neighborhood near $0$. We finally also note here that Theorem~\ref{thm.conv} remains valid when agents data point are not independent but produced by a joint distribution $\mu_{\mathrm{joint}}$. In such a case, the bounds would depend on $E_{\mu_{\mathrm{joint}}}[\frac{1}{n} \sum_i x_ix_i^{\top}]$ (instead of $E_\mu[xx^{\top}]$) which captures precisely the impact of correlation on the estimation cost---we provide details on this in Appendix~\ref{sec:joint}. \looseness=-1

In this degree of generality, it is no longer possible express the equilibrium precisions in closed form (as in Proposition~\ref{thm.scaling}).
Nevertheless, Theorem~\ref{thm.conv} shows that we are able to provide precise bounds for the estimation cost.
In particular, the upper bound in \eqref{eq.conv} shows that, as soon as $\pmin>1$ (i.e., data provision costs are superlinear), the estimation cost converges to zero for any scalarization, meaning that the consistency property of GLS is preserved. If $\pmin  =  1$ though, this is not guaranteed (and even guaranteed to fail if $\pmin  =  \pmax  =  1$, i.e., for linear costs). Even when convergence to zero is guaranteed ($\pmin > 1$), the lower bound in \eqref{eq.conv} shows that the convergence rate is slower that the standard rate of $\Theta(n^{-q})$ (or $\Theta(1/n)$ for scalarizations with $q=1$). \looseness=-1

We immediately see that for the case $\pmax = \pmin = p$, the exponent $\alpha$ is equal to $0$ and the exponents of the left-hand side and of the right hand-side of \eqref{eq.conv} coincide and are equal to the exponent of Proposition~\ref{thm.scaling}. When $\pmin$ and $\pmax$ are different, the bounds loosen. Intuitively, the upper bound then involves the parameter $\pmin$ because, when precisions are close to zero, the agents with exponent $\pmin$ are the ones that have the smallest precision at equilibrium due to larger marginal provision costs. The lower bound, however does not correspond exactly to the $\nbplay^{-\qhomo\frac{\pmax - 1}{\pmax + \qhomo}}$ that one could expect (in fact it decreases faster than $\nbplay^{-\qhomo\frac{\pmax - 1}{\pmax + \qhomo}}$). Whether this is a proof artifact or a consequence of our assumption on the provision costs (which is weak and allows for very diverse costs) remains an open question. We performed a numerical investigation of the result of Theorem~\ref{thm.conv} illustrating the lower and upper bounds---the results are deferred to Appendix~\ref{sec:more_experiments}.

\begin{remark*}
We should also note here that our model formally relies on the GLS estimator---which is based on a principle of truthful revelation of data and of its precision to the analyst.
This is a natural assumption to make for our envisioned applications where agents are motivated by the model's quality.
However, there are other settings where strategic considerations might lead agents to act in a different manner:
For instance, if the agents are rewarded as a function of the precision, they might be tempted to untruthfully disclose a higher precision;
as another example, agents may be unable to properly quantify the precision of their data points.
In such settings, an interesting alternative would be to consider the ordinary least squares (OLS) estimator instead of GLS, as OLS is oblivious to the disclosed precision of the data points. 
In this case, the conclusion of Theorem~\ref{thm.conv} would continue to hold;
due to space limitations, the detailed statement and proof are relegated to Appendix~\ref{sec:ols}.
Our analysis for OLS also reveals a potential shortfall of OLS: a single agent with a high provision cost can cause arbitrarily bad estimation cost (whereas GLS is robust to such agents).
We discuss this in detail in Appendix~\ref{sec:ols}. \looseness=-1
\end{remark*}

\paragraph{Comparison with Theorem~\ref{generalizedBound}}

Theorem~\ref{thm.conv} and Theorem~\ref{generalizedBound} both capture notions of efficiency of the game but they are hardly comparable because they characterize radically different types of inefficiencies. Theorem~\ref{thm.conv} characterizes the ratio of \emph{estimation cost} (the analyst's viewpoint) between the case of strategic agents and a non-strategic scenario where each agent would give a fixed exogenous precision $\lambdans$. In contrast, the $\poa$ of Theorem~\ref{generalizedBound} is a bound of the \emph{total cost} (the population viewpoint) and characterizes the inefficiency due to the self-interested agents by comparing the total cost at equilbrium and at social optimum.  These two situations are radically different and the $\poa$ result of Theorem~\ref{generalizedBound} does not hint at the convergence issues addressed in Theorem~\ref{thm.conv}, even in hindsight. For instance, in the case of linear provision costs ($\pmin=\pmax=1$), $\gls$ is inconsistent whereas Theorem~\ref{generalizedBound} shows that the price of anarchy always grows sublinearly in $n$, even in this case where $\poa \le n^{q/(q+1)}$. \looseness=-1

In addition, the proofs of the two theorems are fundamentally different. The proof of Theorem~\ref{generalizedBound} uses a scaling to transform the social optimum $\lambdab^{\opt}$ into an equilibrium $\lambdabeq$. This approach works because $\lambdab^{\opt}$ and $\lambdabeq$ are respectively the minimizers of the functions $\SC$ and of the potential $\phi$ and because these two functions are tightly related. Such an approach cannot be adapted to start from $\lambdabns$ to obtain an equilibrium $\lambdabeq$ as $\lambdabns$ is not a minimizer. Conversely, the proof of Theorem~\ref{thm.conv} could be adapted to obtain a result in the spirit of Theorem~\ref{generalizedBound} but would lead to a looser bound.

\section{Concluding discussion}
\label{sec:discussion}

In this paper, we show that the precision of GLS estimates for linear regression problems in the presence of strategic data sources is degraded compared to the standard case of non-strategic data sources. We characterize this degradation under mild assumptions. In particular we show that basic properties such as consistency no longer always hold with strategic data sources (if provision costs are linear)---and even when it holds the convergence rate is worsened. This points out a necessity to take into account strategic agents in statistical learning. For example in experiments that have a targeted precision to achieve and want to estimate the number of participants required, our results show that this number should be higher in strategic settings and one should design incentives to have participants with non-linear provision costs. Our work is a stepping stone in this direction.

The objective in our model was to include in the simplest possible way two key elements of learning from strategic data sources: the public good component of the model's precision and the uncertainty about other agents' data. \blue{Our model and results, however, are robust to small changes of assumptions. In particular, we assumed that agents obtain a single data point and that the resulting estimation cost depends on the expected information matrix. Both of these assumptions can be removed from our model at the cost of heavier notation. For instance, we could assume that each agent obtain a set of data points (and not a single datapoint), say distributed according to a probability distribution $\mu_{\mathrm{set}}$. In this case, our main bound of Theorem \ref{thm.conv} would simply depend on $\sum_{S_i \subseteq \Xset} \sum_{x_i \in S_i}x_ix_i^\top \mu_{\mathrm{set}}(S_i)$ instead of $\esp{xx^\top}$.}
We could also assume that agents have a (Bayesian) belief regarding other agents' provision costs. Such an extension would preserve the basic game's structure that leads to the convergence rates of Theorem~\ref{thm.conv}.

In our work, we did not consider incentives that the analyst could give agents to improve the estimation, such as monetary payments. In theory, it would be possible for the analyst to change the provision costs from $\cost_i$ to $\cost^{\prime}_i$ by introducing a payment that is the difference between the two; our results would then directly apply with the new costs $\cost^{\prime}_i$. For instance, if the costs $\cost_i$ are linear and hence GLS is not consistent, it would theoretically be possible for the analyst to select a target convergence rate and enforce it through appropriate new costs $\cost^{\prime}_i$ and the corresponding payments. We note, however, that besides requiring fine knowledge of the actual costs, it would also incur a budget for the analyst that could be arbitrarily large---this is therefore not a valid practical solution.

\acks{
This work was supported by MIAI @ Grenoble Alpes (ANR-19-P3IA-0003), by the French National Research Agency through the ``Investissements d’avenir'' program (ANR-15-IDEX-02 and ANR-11-LABX-0025-01) and through grant ANR-20-CE23-0007; and by the DGA.
}

\clearpage


\bibliography{bibliography,experimental_design_references,privacy_references}
\bibliographystyle{plain}

\numberwithin{lemma}{section}		
\numberwithin{proposition}{section}		
\numberwithin{equation}{section}		
\newpage
\appendix
\onecolumn
\section{Notation table}
\label{sec:notation}

%
%

To ease the reading, Table~\ref{tab:notation} summarizes the main notation introduced in the paper and used throughout the paper and the supplementary material. 

\begin{table}[h!]
\centering
\caption{Summary of the notation}
\label{tab:notation}
 \begin{tabular}{|c|l|} 
 \hline
 Symbol & Meaning \\
 \hline
 $\Xset$ & Finite set of possible attribute vectors $x$. \\ 
 $\proba(\cdot)$ & Probability distribution on attributes vectors $x$.  \\
 $\nbplay$ & Number of agents. \\
 $\strat_i(x)$ & Precision allocated to vector $x$ by player $i$. \\
 $\stratb_{-i}(x)$ & Vector of precisions allocated to vector $x$ by every player except $i$. \\
 $\cost_i(\ell)$ & Data provision cost of agent $i$ for a precision $\ell$ of the data provided.\\
 $\scalar(\info)$ & Scalarization mapping a (covariance) matrix $M$ to a cost.\\
 $\Cpublic{\stratb}$ & Estimation cost $= \scalar\p{\p{\esp{\sum_{i\in N} \lambda_i(x_i)
  x_ix_i^T}}^{-1}}$.\\
 $J_i(\lambda_i,\lambdab_{-i})$ & Payoff of agent $i$ considering the strategy profile $\lambdab = (\lambda_i, \lambdab_{-i})$.\\
 $\potent(\stratb)$ & Potential function of the linear regression game.\\
 $\nu_{\lambdabeq}(x)$ & Measure mapping a vector to its probability $\times$ the sum of precisions attributed by agents.\\
 $p, \pmin, \pmax$ & Homogeneity degrees of provision costs.\\ 
 $\qhomo$ & Homogeneity degree of a scalarization.\\
 $\optdes$ & Optimal design.\\
\hline
\end{tabular}
\end{table}

\section{Scalarizations}
\label{sec:scalarizations}

In this section, we detail some examples of usual matrix scalarizations mentioned briefly in the paper that fit Assumption~\ref{scalarization.type}  and are standard in optimal design. For further information on the subject, see \cite{atkinson2007optimum} and the references therein.

\subsection{Trace}

The trace trivially satisfies Assumption~\ref{scalarization.type} with $q = 1$. It is used in optimal design to minimize the average variance of the estimates of the regression coefficients and is known as the A-optimal design criterion.

\subsection{Squared Frobenius norm}

The squared Frobenius norm is defined on the set of matrices $V = [v_{ij}]$ of dimensions $\ndim \times \ndim$ as:
\begin{align*}
 ||V||_F^2 &= \sum_{i = 1}^\ndim \sum_{j = 1}^\ndim v_{ij}^2\\
 &= \trace(VV^T).
\end{align*}
It is easy to check that this scalarization satisfies Assumption~\ref{scalarization.type} with $q = 2$.

\subsection{Mean squared error}

We define the mean squared error of \emph{an estimator} $\betabh$ estimating a linear model $\betab$ as:
\begin{equation}
 \mse(\betabh) = \esp{(\betabh - \betab)(\betabh - \betab)^T}.
\end{equation}
This mean squared error is simply the estimator's covariance matrix. It is a property of the estimator and it is a classical proxy to assess its quality.\footnote{See Michel F. Dekking, Cornelis Kraaikamp, Paul H. Lopuhaä, and Ludolf Meester. A Modern Introduction to Probability and Statistics: Understanding why and how. Springer Science \& Business Media (2005).} In particular, in the linear regression setting, it does not depend on the realization of the values $\tilde{y}_i$ but only on the independent variables $x_i$ and on the precisions of the response variables $\tilde{y}_i$ (unlike the empirical mean squared error). 

A similar definition can also be applied to the predicted value for a given data point $x$. In this case it is referred to as the mean squared error of the predictor:
\begin{equation*}
  \mse(\betabh^T x) = \esp{(\betabh^T x - \betab^T x)^2}.
\end{equation*}
This quantity gives an indication on the average amount of error the estimator makes when predicting the value of the model on a given data point $x$. It is used in optimal design to define scalarizations by considering the average mean squared error made by the estimator on specific data points. To properly define these criteria, we first write this quantity in a more convenient form.

The mean squared error of the predictor of the linear model on a parameter $x$ is:
\begin{align*}
 \mse(\betabh^T x) = \var(\betabh^T x) + \bias(\betabh^T x, \betab^T x).
\end{align*}
As $\betabh$ is unbiased, we can rewrite the mean-squared error depending only on the variance. Let $\cov$ be the covariance matrix of a linear unbiased estimator $\betabh$. We then have:
\begin{align*}
 \mse(\betabh^T x) &= \var(\betabh^T x) \\
  &= x\cov x^T.
\end{align*}

We now define the two main design criteria (or scalarizations) that are based on this mean squared prediction error:
\begin{enumerate}[a.,leftmargin=22pt,itemsep=0pt,topsep=0pt]
\item \textbf{The average mean squared error.}
Given a set $\mathcal{V}$ and a probability distribution $\rho$ on $\mathcal{V}$, we define the average mean-square error scalarization as:
\begin{equation*}
 \scalar: \cov \rightarrow \int_{\mathcal{V}} x\cov x^T \rho(dx).
\end{equation*}

This scalarization is trivially convex, increasing in the positive semi-definite order and homogeneous of degree $q=1$. It is known in the optimal design litterature as the I (integrated) optimal design criterion and is used to minimize the average prediction error. In our setting this scalarization can be directly applied with $\mathcal{V} = \Xset$ and $\rho = \proba$.

\item \textbf{The mean squared error over a set of specific points.}
Given a finite set $\{x_1, \dots, x_m\}$ of possible attribute vectors, we define the mean-squared error on that specific set of points as: 
\begin{equation*}
 \scalar: \cov \rightarrow \sum_{i = 1}^m x_i \cov x_i^T.
\end{equation*}

This scalarization is similar to the previous one and has the same properties but is used to minimize the prediction error only on a specific set of points of interest. It is known in the optimal design litterature as the V optimal design criteria.

\end{enumerate}


\section{Proofs}
\label{sec:proofs}
%
%

\subsection{Proof of Proposition~\ref{thm.uniqueness}}
\label{proof.uniqueness}

Recall that a strategy $\lambda$ is a function from the finite set $\Xset$ to $\R_+$. Hence, a strategy $\lambda$ is an element of the finite dimensional space $\R^{\Xset}$ and a precision profile $\lambdab$ is essentially a vector (of dimension $n|\Xset|$).

\textbf{Step 1: The potential function is convex.} By Assumption~\ref{privacy.type}, the data provision costs are convex. Additionally, $\Cpublic{\stratb}$ is a composition of the function $\stratb \rightarrow \esp{\sum_i \strat_i(x_i)x_ix_i^T}$, the matrix inverse function, and the scalarization. The matrix inverse function is a convex function. As the scalarization $F$ is non-decreasing and convex (by Assumption~\ref{scalarization.type}), $\Cpublic{\stratb}$ is convex in $\stratb$. This shows that the potential is convex; hence a strategy profile is a Nash equilibrium if and only if it is a minimum of the potential.

\textbf{Step 2: The potential admits a minimum.} Let $\potent(\mathbf{1}) = \sum_i \esp{c_i(1)} + \scalar((\esp{\sum_i x_ix_i^T})^{-1})$. By Assumption~\ref{privacy.type}, $\lim_{\ell \to +\infty}\cost_i(\ell) = +\infty$. Recall that $\mu$ has full support on $\Xset$ (Assumption~\ref{structure.information}). Then, for all $x\in\Xset$, $\lim_{\ell \to +\infty}\cost_i(\ell)\mu(x) = +\infty$. Hence, there exists $\ellmax$ such that for all $i$ and all $x$, $c_i(\ellmax)\mu(x) > \potent(\mathbf{1})$. This shows that if $\lambdab$ is a precision profile such that $\lambda_i(x)>\ellmax$ for some $i$ and $x$, then $\phi(\lambdab)\ge\phi(1)$.

As $\Xset$ is finite, the set of precision profile such that for all $i$, $\strat_i: \Xset \rightarrow [0, \ellmax]$ is a compact set. As $\potent$ is convex, it admits a minimum on this set. By definition of $\ellmax$, this minimum is a global minimum. This concludes the proof that there exists an equilibrium. If in addition all data provision costs are strictly convex, then the potential is strictly convex; hence this minimum is unique and there exists a unique equilibrium.

\textbf{Step 3: If different equilibria exist, they have the same estimation cost.}
As shown before, an equilibrium is a minimum of the potential function $\phi$ defined for all precision profiles $\lambdab$ as 
\begin{align*}
  \phi(\lambdab)=\sum_{i}\esp{c_i(\lambda_i(x))}+\Cpublic{\lambdab}.
\end{align*}
In the above equation, $\Cpublic{}$ is not necessarily strictly
convex. Recall indeed that $\Cpublic{\lambdab}$ is defined as
\begin{align*}
  \Cpublic{\lambdab} = \scalar\Bigg(\p{\esp{\sum_{i} \lambda_i(x_i) x_ix_i^T}}^{-1}\Bigg).
\end{align*}
If there exist $\lambdab\ne\lambdab'$ (which is the case for any linear regression game with $n\ge2$ players) such that $\esp{\sum_{i} \lambda_i(x_i) x_ix_i^T}=\esp{\sum_{i} \lambda'_i(x_i) x_ix_i^T}$, then $\Cpublic{\lambdab}=\Cpublic{\lambdab'}=\Cpublic{(\lambdab+\lambdab')/2}$ and $\Cpublic{}$ is not strictly convex.

Yet, we show below that $\Cpublic{\cdot}$ is strictly convex when viewed as a function of $M(\lambdab)=\esp{\sum_{i} \lambda_i(x_i) x_ix_i^T}$. Indeed $F$ is an increasing convex function (by Assumption~\ref{scalarization.type}, see erratum in Appendix~\ref{sec:notation}) and $M\mapsto M^{-1}$ is a strictly convex function, the function $M\mapsto F(M^{-1})$ is a strictly convex function.

Assume that there exist two equilibria $\stratbeq$ and $\stratbeqt$ and assume by contradiction that $\esp{\sum_i \strateq_i(x_i)x_ix_i^T} \neq \esp{\sum_i \strateqt_i(x_i)x_ix_i^T}$. Let $\stratb^\prime=(\stratbeq+\stratbeqt)/2$. The strict convexity of $M\mapsto F(M^{-1})$ implies that $\Cpublic{\stratb^\prime}<(\Cpublic{\stratbeq}+\Cpublic{\stratbeqt})/2$. This implies that $\potent(\stratb^\prime) < (\potent(\stratbeq)+\potent(\stratbeqt))/2$, which contradicts the fact that $\stratbeq$ and $\stratbeqt$ are minima of the potential function $\phi$. Thus, if two different equilibria exist, they have the same information matrix and yield the same estimation cost.

\subsection{Proof of Theorem~\ref{generalizedBound}}

This proof relies on adapting the proof of \cite{Gast20a} to our setting. For completeness, we redo this proof using the notations of our model.

\medskip

\noindent\textbf{Upper Bound.} 
To simplify the notation, in this proof, we write $p$ instead of $\pmin$; hence we show that $\poa \le n^{\frac{q}{p + q}}$.
Suppose that $\poa > n^{\frac{q}{p + q}}$. 
This implies that there exists an equilibrium $\lambdabeq$ such that
\begin{align*}
  \SC(\lambdabeq)&\ge \sum_{i\in N} \esp{c_i(\lambdaeq_i(x_i))} + n\Cpublic{\lambdabeq} \\
  &> n^{\frac{q}{q + p}} (\sum_{i\in N} \esp{c_i(\lambda_i^\opt(x_i)} + n\Cpublic{\lambdab^\opt})\\
  &= n^{\frac{q}{q + p}}\SC({\lambdab}^{\opt}).
\end{align*}
We will show that this implies that $\lambdabeq$ is not an equilibrium, which is a contradiction. 

By using that $c_i(\lambdaeq_i)\ge 0$ and dividing the above inequality
by $n$, we obtain:
\begin{align*}
  \Phi(\lambdabeq)& = \sum_{i\in N} \esp{c_i(\lambdaeq_i(x_i))} + \Cpublic{\lambdabeq}\\
  &\ge \frac1n\left(\sum_{i\in N} \esp{c_i(\lambdaeq_i(x_i))} + n\Cpublic{\lambdabeq}\right) 
  & (= \frac1n\SC(\lambdabeq))\\
   &> n^{-\frac{p}{q + p}}\sum_{i\in N} \esp{c_i(\lambda_i^\opt(x_i))} + n^{\frac{q}{p + q}}\Cpublic{\lambdab^\opt}
   &(= \frac1n n^{\frac{q}{q + p}}\SC({\lambdab}^{\opt}))\\
  &\ge \sum_{i\in N} \esp{c_i\left(\frac{\lambda_i^\opt(x_i)}{n^{\frac{1}{p + q}}}\right)} + 
    \Cpublic{\frac{\lambdab^\opt}{n^{\frac{1}{p + q}}}}& (=\Phi(\lambdab^\opt/(n^{1/(p+q)}))),
\end{align*}
where we used the homogeneity assumptions for the last inequality.

To conclude the proof, we remark that
$\frac{\lambdab^\opt}{n^{1/(p + q)}}$ is a valid strategy profile. This
would imply that $\lambdabeq$ is not a minimum of the potential
function, which is a contradiction. Thus, we have
$\poa \le n^{\frac{q}{p + q}}$.


\medskip

\noindent\textbf{Lower Bound.} Let $p,q\ge1$. Consider the linear regression game where $\Xset = \{1\}$, $\mu(1) = 1$, $c_i(\ell)=\ell^p$ and $F(V)=\mathrm{trace}(V)^q=V^q$. As $\mu$ is a deterministic measure, this game is also a valid game in the setting of \cite{Gast20a}. It is straightforward to see that our game has a unique Nash equilibrium $\lambdabeq$ that corresponds to the unique non-trivial Nash equilibrium of the corresponding game of \cite{Gast20a}. Hence, the price of anarchy of our game coincides with the price of stability of the corresponding game of \cite{Gast20a}. Hence, the computation of \cite{Gast20a} show that, for all $\epsilon$, there exist $n$ such that this game has a price of anarchy larger than $n^{q/{p+q}}(1-\epsilon$). 

\subsection{Proof of Theorem~\ref{thm.optim.design}}
\label{proof.optim.design}

Recall that the provision cost of a player $i$ is
$c_i(\ell)=a_i\ell$ and assume without loss of generality that $a_1\le
a_2\le \dots \le a_n$.


Let $\lambdabeq$ be an equilibrium of the game and let $\optdes$ be an optimal design. Recall that $\nu_{\lambdabeq}(x)=\sum_{i\in N}\lambdaeq_i(x)\mu(x)$ for all $x\in \Xset$. Let $b=\sum_{x\in\Xset}\nu_{\lambdabeq}(x)$.  Let $\lambda_{\optdes}$ be the strategy such that $\lambda_{\optdes}(x)=b\optdes(x)/\mu(x)$ for all $x$ and consider the precision profile $\lambdab_{\optdes}=(\lambda_{\optdes},0,\cdots,0)$. We have:
\begin{align}
  \phi(\lambdabeq) &= F((\sum_xxx^T\nu_{\lambdabeq}(x))^{-1}) +
                     \sum_i a_i\sum_x\lambdaeq_i(x)\mu(x)\nonumber\\
                   &\ge F( (\sum_xxx^T\nu_{\lambdabeq}(x))^{-1} )
                     + a_1b \label{eq:ineq1}\\
                   &= b^{-q}F( (\sum_xxx^T\nu_{\lambdabeq}(x)/b)^{-1} )
                      + a_1b  \label{eq:ega1}\\
                   &\ge b^{-q}F( (\sum_xxx^T\optdes(x))^{-1} ) +
                     a_1b \label{eq:ineq2}\\
                   &=F((\sum_xxx^T\lambda_{\optdes}(x)\mu(x))^{-1}) +
                     a_1\sum_x\lambda_{\optdes}(x)\mu(x)\label{eq:ega2}\\
                   &=  \phi(\lambdab_{\optdes}),\nonumber
\end{align}
where the first inequality \eqref{eq:ineq1} is because $a_1\le a_i$ for all $i$, and 
the second inequality \eqref{eq:ineq2} is because $\nu^*$ is an
optimal design. The equalities \eqref{eq:ega1} and \eqref{eq:ega2} are
due to the homogeneity of $F$ (Assumption~\ref{scalarization.type}
implies that $F((bM)^{-1})=b^{-q}F(M^{-1})$), and in \eqref{eq:ega2} we also use that by definition of $\lambda_{\optdes}$ and since $\sum_x \optdes = 1$ we have $\sum_x\lambda_{\optdes}(x)\mu(x) = b$.

If $\nu_{\lambdabeq}/b$ was not an optimal design, the inequality \eqref{eq:ineq2} would be strict which would imply that $\phi(\lambdabeq)>\phi(\lambdab_{\optdes})$ which would contradict the fact that $\lambdabeq$ is a minimum of the potential. This implies that \eqref{eq:ineq2} is an equality which means that $\nu_{\lambdabeq}(x)/b$ is an optimal design.

\subsection{Proof of Proposition~\ref{thm.scaling}}
\label{proof.scaling}

An equilibrium is a minimum of the potential function $\phi$. When all costs are identical, this function is symmetric. As $\phi$ is a convex function, this implies that there exists a minimum of $\phi$ that is symmetric.  A symmetric precision profile $\lambdab=(\lambda,\dots\lambda)$ is a Nash equilibrium if and only if it minimizes the potential $\phi$. By symmetry, this potential can be rewritten as:
\begin{align*}
  \phi(\lambda,\dots,\lambda) &= n \esp{\lambda(x)^p}+
                                \Cpublic{n\lambda}
\end{align*}
Let us define the function $f:\R_+^{\Xset}\to\R_+$ that associates
to a strategy $\lambda$, the quantity
$f(\lambda)=\esp{\lambda(x)^p} + \Cpublic{\lambda}$. Recall that
$\stratsingle$ is the minimum of $f$. For a given strategy $\lambda$,
we have:
\begin{align*}
  \phi(\nbplay^{-\frac{\qhomo + 1}{\phomo +
  \qhomo}}\lambda,\dots,\nbplay^{-\frac{\qhomo + 1}{\phomo +
  \qhomo}}\lambda) 
  &=n \esp{\lambda(x)^pn^{-\frac{q+1}{p+q}p}}+
    \Cpublic{nn^{-\frac{q+1}{p+q}}\lambda}\\
  &= n^{q\frac{1-p}{p+q}} \esp{\lambda(x)^p} +
    n^{-q\frac{p-1}{p+q}}\Cpublic{\lambda} 
  \\
  &= \nbplay^{-\qhomo\frac{\phomo - 1}{\phomo + \qhomo}} \Cpublic{\lambda},
\end{align*}
where we used the homogeneity of $F$, which implies that
$\Cpublic{a\lambda}=a^{-q}\Cpublic{\lambda}$. 

For any $n\in\{1,2,\dots\}$, the function $\lambda \mapsto \nbplay^{-\frac{\qhomo + 1}{\phomo + \qhomo}} \lambda$ is a bijection from $\R_+^\Xset$ to $\R_+^\Xset$. Hence, $\lambda$ is a minimum of $f$ if and only if $(\nbplay^{-\frac{\qhomo + 1}{\phomo + \qhomo}}\lambda,\dots,\nbplay^{-\frac{\qhomo + 1}{\phomo + \qhomo}}\lambda)$ is a minimum of $\potent$.  Thus, the precision profile $\stratbeq$ such that $\forall i: \lambdaeq_i=n^{-\frac{1+q}{p+q}}\stratsingle$ is an equilibrium.

The second part of the proposition follows immediately from the homogeneity of $F$, which implies that for this equilibrium, $\Cpublic{\stratbeq}=n^{-q\frac{p-1}{p+q}}\Cpublic{\stratsingle}$. Moreover, all equilibria have the same estimation cost by Proposition~\ref{thm.uniqueness}.


\subsection{Proof of Theorem~\ref{thm.conv}}
\label{proof.conv}
\subsubsection{Upper bound}

In this first step, we compute the value of the potential function for a particular constant strategy in which all players use the precision $\lambda(x)=\nbplay^{-\frac{\qhomo + 1}{\pmin + \qhomo}}$ for all values of $x\in\Xset$.  By abuse of notation, we denote this precision profile by  $(\nbplay^{-\frac{\qhomo + 1}{\pmin + \qhomo}},\dots,\nbplay^{-\frac{\qhomo + 1}{\pmin + \qhomo}})$. The value of the potential for this precision profile is
\begin{align}
  \potent(\nbplay^{-\frac{\qhomo + 1}{\pmin +
  \qhomo}},\dots,\nbplay^{-\frac{\qhomo + 1}{\pmin + \qhomo}})
  &= \sum_{\play = 1}^{\nbplay} \esp{\cost_\play(\nbplay^{-\frac{\qhomo + 1}{\pmin + \qhomo}})} + \scalar((\sum_{\play = 1}^{\nbplay} \esp{xx^T \nbplay^{-\frac{\qhomo + 1}{\pmin + \qhomo}}})^{-1})\nonumber\\
  &=\sum_{\play = 1}^{\nbplay}\cost_\play(\nbplay^{-\frac{\qhomo + 1}{\pmin + \qhomo}})
    +\scalar((n^{\frac{\pmin-1}{\pmin + \qhomo}} \esp{xx^T})^{-1})\nonumber\\
  & \le \sum_{\play = 1}^\nbplay \nbplay^{-\pmin\frac{\qhomo + 1}{\pmin + q}} \cost_i(1) + \scalar((\nbplay^{\frac{\pmin-1}{\pmin + q}} \esp{xx^T})^{-1})\label{eq:line2}\\
  &=\nbplay^{-\pmin\frac{\qhomo + 1}{\pmin + q}} \sum_{\play = 1}^\nbplay \cost_i(1) + \nbplay^{\frac{\qhomo(1 - \pmin)}{\pmin + \qhomo}}\scalar((\esp{xx^T})^{-1})\label{eq:line2bis}\\
  &\le \nbplay^{-\frac{\qhomo(\pmin-1)}{\pmin + \qhomo}}c_{\max}(1) + \nbplay^{\frac{\qhomo(1 - \pmin)}{\pmin + \qhomo}}\scalar((\esp{xx^T})^{-1})\label{eq:line3}\\
  &= \nbplay^{-\frac{\qhomo(\pmin-1)}{\pmin + \qhomo}}\p{c_{\max}(1) + \scalar((\esp{xx^T})^{-1})}\label{eq:line4},
 \end{align}
 where we use that $c_i(1)\ge a^{\pmin}c_i(1/a)$ with $a=n^{\frac{q+1}{\pmin+q}}$ (from the theorem's assumption) in \eqref{eq:line2}, the homogeneity of $F$ (Assumption~\ref{scalarization.type}) in \eqref{eq:line2bis}, and the theorem's assumption, which implies that $c_i(1)\le c_{\max}(1)$ for all $i$, in \eqref{eq:line3}.

 As $c_i(\ell)\ge0$ and $\lambdabeq$ is a minimum of the potential, it
 holds that
 \begin{align*}
   \Cpublic{\lambdabeq}&\le\potent(\lambdabeq)\le\potent(\nbplay^{-\frac{\qhomo + 1}{\pmin +
  \qhomo}},\dots,\nbplay^{-\frac{\qhomo + 1}{\pmin + \qhomo}}).
 \end{align*}
 Hence, the right-hand-side of \eqref{eq.conv} holds with $D=\p{c_{\max}(1) + \scalar((\esp{xx^T})^{-1})}$. 

\subsubsection{Lower bound}

By \eqref{eq:line4}, $\phi(\lambdabeq)\le \nbplay^{-\frac{\qhomo(\pmin-1)}{\pmin + \qhomo}} (c_{\max}(1) + \scalar((\esp{xx^T})^{-1}))$. Recall
that all $c_i$ are increasing convex and
$\inf_i c_i(1)\ge c_{\min}(1)>0$. This implies that
$\lim_{\ell\to\infty}\inf_ic_{i}(\ell)=\infty$ as
$\inf_ic_{i}(\ell)>\ell^{\pmin}c_{\min}(1)$. As $\lambdabeq$ is a
minimum of the potential, this implies that there exists a value
$\ellmax$ independent of $n$ such that $\lambdaeq_i(x)\le\ellmax$. 

 We first obtain a bound on the total amount of precision given by all
 players. To do that we use Jensen's inequality for concave function
 in \eqref{eq:step2_1}. Then we use that $c_i(\ellmax)\le
 (\ellmax/\lambda_i(x))^{\pmax}c_i(\lambda_i(x))$ as 
 $\ellmax/\lambda_i(x)>1$ to obtain \eqref{eq:step2_2} and $c_i(\ellmax)\ge
 c_{\min}(\ellmax)$ to obtain \eqref{eq:step2_3}:
 \begin{align}
   \p{\sum_{\play = 1}^\nbplay \frac{1}{\nbplay} \esp{\cost_\play(\strat_\play(x))}}^{\frac{1}{\pmax}} & \ge \sum_{\play = 1}^\nbplay \frac{1}{\nbplay} \esp{\p{\cost_\play(\strat_\play(x))}^{\frac{1}{\pmax}}}\label{eq:step2_1}\\
  &\ge \sum_{\play = 1}^\nbplay  \frac{1}{\nbplay}\esp{((\strat_\play(x)/\ellmax)^{\pmax}\cost_\play(\ellmax))^\frac{1}{\pmax}} \label{eq:step2_2}\\
  &\ge \frac{(c_{\min}(\ellmax))^{\frac{1}{\pmax}}}{\ellmax} \frac{1}{\nbplay} \sum_{\play = 1}^\nbplay \esp{\strat_\play(x)}.\label{eq:step2_3}
 \end{align}

This shows that %
 \begin{align}
   \sum_{\play = 1}^\nbplay \esp{ \strateq_\play(x)}
   &\le\frac{n\ellmax}{(c_{\min}(\ellmax))^{1/\pmax}} \p{\sum_{\play = 1}^\nbplay
     \frac{1}{\nbplay}
     \esp{\cost_\play(\strateq_\play(x))}}^{\frac{1}{\pmax}} \nonumber\\
   &\le\frac{n\ellmax}{(c_{\min}(\ellmax))^{1/\pmax}} \p{\frac1n\phi(\lambdabeq)}^{1/\pmax}\nonumber\\
     &\le
       \frac{n\ellmax}{(c_{\min}(\ellmax))^{1/\pmax}}\p{\frac1n\nbplay^{-\frac{\qhomo(\pmin-1)}{\pmin
       + \qhomo}}\p{c_{\max}(\ellmax) + \scalar((\esp{xx^T})^{-1})}}^{1/\pmax},
       \label{eq:step2_4}
 \end{align}
 where we used \eqref{eq:step2_3} for the first inequality, the fact
 that $\Cpublic{\lambdab}\ge0$ for the second and \eqref{eq:line4} to obtain
 the last inequality. 

 Note that the exponent of $n$ in \eqref{eq:step2_4} is 
 \begin{align*}
   1-1/\pmax-\frac{q(\pmin-1)}{\pmax(\pmin+q)}
   &=\frac{\pmax(\pmin+q)-(\pmin+q)-q(\pmin-1)}{\pmax(\pmin+q)}\\
   &=\frac{\pmax(\pmin+q)-\pmin(1+q)}{\pmax(\pmin+q)}\\
   &=\frac{\pmax(\pmin-1)+(\pmax-\pmin)(1+q)}{\pmax(\pmin+q)}\\
   &=\frac{\pmin-1}{\pmin+q} + \alpha/q,
 \end{align*}
 where $\alpha=\qhomo\frac{(\pmax-\pmin)(q+1)}{\pmax(q+\pmin)}$ is the 
 same $\alpha$ as in Theorem~\ref{thm.conv}.
 
Plugging this into \eqref{eq:step2_4} yields the upper bound on the total amount of precision given by all
 players:
 \begin{equation}
 \label{precision.major}
  \sum_{\play = 1}^\nbplay\esp{\strateq_\play(x)} \le \ellmax\p{1 +
  \frac{\scalar((\esp{xx^T})^{-1})}{c_{\min}(\ellmax)}}^{\frac{1}{\pmax}} \nbplay^{\frac{\pmin-1}{\pmin+q}+\alpha/q}.
 \end{equation}

 Recall that $\nu_{\lambdabeq}(x)=\sum_i\lambda_i(x)\mu(x)$.
 Following what we did in \eqref{eq:ineq2} with the notation
 $b=\sum_{x\in\Xset}\nu_{\lambdabeq}(x)=\esp{\sum_{i}\lambdaeq_i(x)}$,
 we have
 \begin{align}
   \Cpublic{\lambdabeq} &\ge \p{\esp{\sum_{i}\lambdaeq_i(x)}}^{-q}
                          F\p{\p{\sum_{x\in\Xset}
                          xx^T\optdes(x)}^{-1}}.
                          \label{eq:optdes_proof}
 \end{align}
 Combining \eqref{eq:optdes_proof} and \eqref{precision.major} shows that the
 right-hand-side of \eqref{eq.conv} holds with 
 \begin{align*}
   d=F\p{\p{\sum_{x\in\Xset}
   xx^T\optdes(x)}^{-1}}\ellmax\p{1 +
  \frac{\scalar((\esp{xx^T})^{-1})}{c_{\min}(\ellmax)}}^{-\frac{q}{\pmax}}.
 \end{align*}


%

\section{Equivalence}
\label{sec:equiv}

In this section, we show that our model is equivalent to the complete information model defined in \cite{Gast20a}, when the number of player goes to infinity. We consider a model with $n$ agents in which the feature of agent $i$ are chosen \emph{i.i.d.}. The only difference between the two models is that: 
\begin{itemize}
  \item In our model, an agent $i$ does not know the exact feature $x_{-i}$ of the other individual but only knows the distribution $\mu$ from which they are drawn. As a result, an player $i$ seeks to minimize
  \begin{align*}
    J_i(\lambda_i,\lambdab_{-i})=\esp{c_i(\lambda_i(x_i))} + \scalar\Bigg(\p{\esp{\sum_{i\in N} \lambda_i(x_i) x_ix_i^{\top}}}^{-1}\Bigg),
  \end{align*}
  where $\lambda_i:\Xset\to\R^+$ is a function that associates to each possible feature $x\in\Xset$ a precision $\lambda_i(x)$. 
  \item In the model of \cite{Gast20a}, a player $i$ knows the exact features of other players. As a result, its cost function is
  \begin{equation}
    \label{payoff.complete}
    \payoffci_i(\ell_i, \ell_{-i}, X) = c_i(\ell_i) + \scalar((\sum_{i = 1}^n \ell_ix_ix_i^T)^{-1}).
  \end{equation}
  In the above definition, we emphasize that the cost of a player depends on $X=(x_i)_{i\in N}$ which is the matrix of features of all players. In particular, the equilibrium of the complete information game is well defined only if $\sum_i x_i x_i^T \succ 0$. This assumption simply states that the data points held by the agents span $\R^\ndim$ so that the corresponding linear regression is well defined. 
We refer to \cite{Gast20a} for technical results regarding the existence of a Nash equilibrium. The complete information game is a potential game with the  potential:
\begin{equation}
 \label{potential.complete}
 \potentci(\bell, X) = \sum_{i = 1}^nc_i(\ell_i) + \scalar((\sum_{i = 1}^n \ell_ix_ix_i^T)^{-1}),
\end{equation}
and the Nash equilibrium of the game is the minimum of the potential function.
\end{itemize}

In this section, we show that when $n \rightarrow +\infty$, the equilibrium of the complete information game, that we denote by $\bell^{\ci*}$, and the equilibrium of our linear regression $\stratbeq$ are equivalent and can be exchanged.

\textbf{Notations and assumptions:} We assume the same assumptions as Theorem~\ref{thm.conv}. In addition, we assume that there is a finite number $T$ of provision cost functions and we denote by $n_t$ the number of agents having provision cost $c_t$ for $t\in\types:=\{1\dots T\}$.

\subsection{Comparison of equilibrium} 

To formally compare the equilibrium of the complete information game to the equilibrium of our linear regression game, we will use need the following lemma. This lemma states that there always exists a symmetric equilibrium of the games considered. Note that if provision costs are strictly convex, the equilibrium is unique. If provision costs are linear,  there might, however, exist an infinite number of equilibrium. 

\begin{lemma}
  \label{lem.eql}
  There exists an equilibrium of the complete information game $\ell^{\ci*}$ such that:
  \begin{equation}
  \forall i,i^\prime, x_i = x_{i^\prime} \text{ and } c_i = c_{i^\prime}\Rightarrow \ell^{\ci*}_i = \ell^{\ci*}_{i^\prime}
  \end{equation}
  There exists an equilibrium of the linear regression game $\stratbeq$ such that:
  \begin{equation}
  \forall i, i^\prime, c_i=c_{i^\prime} \Rightarrow \forall x \in \Xset, \strateq_i(x) = \strateq_{i^\prime}(x) 
  \end{equation}
\end{lemma}

\begin{proof}
  Consider an equilibrium $\bell^{\ci*}$ of the complete information game. We define the following strategy profile:
  \begin{equation*}
  \forall i \in N, \ell_i = \sum_{i^\prime = 1}^n \mathbbm{1}_{c_i = c_t \text{ and } x_{i^\prime = x_i}} \frac{\ell^{\ci*}_{i^\prime}}{n_t^{x_i}},
  \end{equation*}
  where $n_t^{x_i}$ is the number of players with features $x_i$ and cost type $t$. 

  This strategy profile is simply that each agent provides data with the precision being the average of the precision of similar agents in the equilibrium. It achieves the same estimation cost as the equilibrium and with our convexity assumptions achieves a lower total provision cost. This is thus a minimum of the potential and an equilibrium.

  The proof for the linear regression game follows the same steps.
\end{proof}

As there is a symmetric equilbrium, this implies that instead of considering strategy profiles, we may restrict our attention to functions $\strat_t(x)$ that associate a type of cost and a data point to a precision.  This is true for the Bayesian game, in which $\strat_i(x)$ is replaced by $\strat_t(x)$ when $c_i=c_t$. This is also true for the complete information game, when $\ell_i$ is replaced by $\strat_t(x_i)$ when $c_i=c_t$.  We work with these functions for the rest of the section and by abuse of notation we redefine the potential of the games as follows:
\begin{align}
  \potentci(\stratb, X) &= \sum_{x \in \Xset}\sum_{t = 1}^T  c_t(\strat_t(x)) n_t^x + \scalar((\sum_{x \in \Xset} xx^T \sum_{t = 1}^T \strat_t(x)n_t^x)^{-1})
  \label{eq:equiv-phici}\\
  \potent(\stratb) &= \sum_{x \in \Xset}\sum_{t = 1}^T  c_t(\strat_t(x)) n_t\proba(x) + \scalar((\sum_{x \in \Xset} xx^T \sum_{t = 1}^T \strat_t(x)n_t\proba(x))^{-1}),
  \label{eq:equiv-phi}
\end{align}
where as before, $n_t$ is the number of players having cost function $c_t$ and $n_t^x$ is the number of player having cost function $c_t$ and features $x$ in the complete information game. 

By abuse of notation, we write $\stratbeq=(\strateq_t)_{t\in\types}$ the equilibrium of our linear regression game and by $\stratbeqci=(\strateqci_t)_{t\in\types}$ the equilibrium of the complete information game. They are the minimum of (respectively) the potential functions \eqref{eq:equiv-phici} and \eqref{eq:equiv-phi}.

\subsection{Main equivalence result}


The intuition behind the theorem is as follows. Equation~\eqref{eq.pot.equiv} states that the minimum of the potentials are equivalent with high probability. Thus, computing the equilibrium of our linear regression game gives a general result on how large complete information games behave. Equations~\eqref{eq.exchange1} and \eqref{eq.exchange2} state that the equilibrium are essentially equivalent. This means that agents can safely compute the equilibrium of the linear regression game without needing to acquire the information of all other agents. We remark that \eqref{eq.pot.equiv} applied with $\pmax = 1$ yields $\potent(\strat^*) = \potentci(\strat^{\ci*})$. Finally, we emphasize that the complexity of Theorem~\ref{thm.equiv} comes from the necessity to prove \emph{equivalence} of potential to show that our results are also valid for the complete information game. Indeed, it is easy to show that both potential go to $0$ as long as $\pmin > 1$. Thus, any result simply stating that the potential of the complete information game converges to the potential of our model is meaningless. With our result, however, it is easy to show that Theorem~\ref{thm.conv} is valid in the complete information setting with high probability.

\begin{theorem}
  \label{thm.equiv}
  Let $\stratbeq$ be an equilibrium of the linear regression game and $\stratbeqci$ be an equilibrium of the complete information game. For all $ 0 < \epsilon < 1/2$, we have with probability at least $1 -|X|\sum_t 2 \exp(-2n_t^{2\epsilon})$:
  \begin{equation}
    \label{eq.pot.equiv}
    \frac{1}{\max_{x, t}\left(\frac{\proba(x) + n_t^{\epsilon - 1/2}}{\proba(x) }\right)^{\pmax - 1}}\potent(\strat^*) \le \potentci(\stratbeqci, X) \le \max_{x, t}\left(\frac{\proba(x)}{\proba(x) - n_t^{\epsilon - 1/2}}\right)^{\pmax - 1} \potent(\strat^*),
  \end{equation}
  \begin{equation}
    \label{eq.exchange1}
    \potentci(\stratbeq, X) \le D_n \max_{x, t}(\frac{\proba(x) + n_t^{\epsilon - 1/2}}{\proba(x) })^{\pmax - 1} \potentci(\stratbeqci, X),
  \end{equation}
  and
  \begin{equation}
    \label{eq.exchange2}
    \potent(\stratbeqci) \le D_n^\prime \max_{x, t}(\frac{\proba(x)}{\proba(x) - n_t^{\epsilon - 1/2}})^{\pmax - 1}\potent(\stratbeq);
  \end{equation}
where
  \begin{align*}
    &D_n = \max(\max_{x,t}(\frac{\proba(x) + n_t^{\epsilon - 1/2}}{\proba(x)n_t}), \frac{1}{(\min_{x,t}(\frac{\proba(x)}{\proba(x) - n_t^{\epsilon - 1/2}}))^q}) \quad\text{ and }\\
    & D_n^\prime = \max(\max_{x,t}(\frac{\proba(x)}{\proba(x) - n_t^{\epsilon - 1/2}}), \frac{1}{(\min_{x,t}(\frac{\proba(x)}{\proba(x) + n_t^{\epsilon - 1/2}}))^q}).
  \end{align*}
\end{theorem}

\begin{proof}
  The equilibrium are defined as $\stratbeqci \in \arg\min(\potentci(\stratb, X))$ and $\stratbeq \in \arg\min(\potent(\stratb))$, where the potential functions are defined in Equations \eqref{eq:equiv-phici} and \eqref{eq:equiv-phi}. 

We define $\stratbeqt(x) = \stratbeq(x)\frac{\proba(x)n_t}{n_t^x}$. As $\stratbeqci$ attains the minimum of $\potentci$, we have:
\begin{align}
  \potentci(\stratbeqci, X) &\le \potentci(\stratbeqt, X)\nonumber\\
  &= \sum_x\sum_t  c_t(\strat_t^*(x)\frac{\proba(x)n_t}{n_t^x}) n_t^x + F((\sum_x xx^T\sum_t \strat_t^*(x) \proba(x))^{-1})\nonumber\\
  & \le \sum_x \sum_t (\frac{\proba(x)n_t}{n_t^x})^\pmax c_t(\strat_t^*(x)) n_t^x + F((\sum_x xx^T\sum_t \strat_t^*(x) \proba(x))^{-1})\label{eq:equiv-ineq2}\\
  & = \sum_x \sum_t (\frac{\proba(x)n_t}{n_t^x})^{\pmax - 1} c_t(\strat_t^*(x)) n_t^x \proba(x) + F((\sum_x xx^T\sum_t \strat_t^*(x) \proba(x))^{-1})\nonumber\\
  & \le \max_{x, t}(\frac{\proba(x)n_t}{n_t^x})^{\pmax - 1} \potent(\stratbeq),\label{eq:equiv-ineq3}
\end{align}
where the inequality \eqref{eq:equiv-ineq2} comes from the assumption on the costs and the inequality \eqref{eq:equiv-ineq3} comes from the fact that $\max_x (\frac{\proba(x)n_t}{n_t^x}) \ge 1$ (Indeed, we have by definition $\sum_x n_t^x = n_t = \sum_x \proba(x)n_t$. Thus, there exists $x \in \Xset$ such that $n_t^x \ge \proba(x)n_t$).

We can prove similarly that:
\begin{equation*}
 \potent(\stratbeq) \le \max_{x, t}(\frac{n_t^x}{\proba(x)n_t})^{\pmax - 1} \potentci(\stratbeqci)
\end{equation*}

We thus obtain that:
\begin{equation}
  \label{eq.equiv}
  \frac{1}{\max_{x, t}(\frac{n_t^x}{\proba(x)n_t})^{\pmax - 1}}\potent(\stratbeq) \le \potentci(\stratbeqci, X) \le \max_{x, t}(\frac{\proba(x)n_t}{n_t^x})^{\pmax - 1} \potent(\stratbeq)
\end{equation}

\noindent \textbf{High probability bound on $\frac{\proba(x) n_t}{n_t^x}$}

Hoeffding inequality implies that for all $t, x$, we have $P(|n_t^x - n_t\proba(x)| \ge k) \le 2 \exp(-\frac{2k^2}{n_t^2})$. We apply this with $k = n_t^{1/2 + \epsilon}$ for $0 < \epsilon < 1/2$ to obtain:
\begin{equation}
  P(|n_t^x - n_t\proba(x)| \ge n_t^{1/2 + \epsilon}) \le 2 \exp(-2n_t^{2\epsilon})
\end{equation}

We thus have $P(\cup_{t, x} (|n_t^x - n_t\proba(x)| \ge n_t^{1/2 + \epsilon})) \le |X|\sum_t 2 \exp(-2n_t^{2\epsilon})$.
We also note that if we have $|n_t^x - n_t\proba(x)| \le n_t^{1/2 + \epsilon}$, then:
\begin{equation*}
 \frac{\proba(x)n_t}{n_t\proba(x) + n_t^{1/2 + \epsilon}}\le \frac{\proba(x)n_t}{n_t^x} \le \frac{\proba(x)n_t}{n_t\proba(x) - n_t^{1/2 + \epsilon}},
\end{equation*}
which yields:
\begin{equation}
  \label{eq.proba.equiv}
  \frac{\proba(x)}{\proba(x) + n_t^{\epsilon - 1/2}}\le \frac{\proba(x)n_t}{n_t^x} \le \frac{\proba(x)}{\proba(x) - n_t^{\epsilon - 1/2}}.
\end{equation}

Combined with \eqref{eq.equiv}, this shows that with probability at least $|X|\sum_t 2 \exp(-2n_t^{2\epsilon})$, we have:
\begin{equation*}
  \frac{1}{\max_{x, t}(\frac{\proba(x) + n_t^{\epsilon - 1/2}}{\proba(x) })^{\pmax - 1}}\potent(\stratbeq) \le \potentci(\stratbeqci, X) \le \max_{x, t}(\frac{\proba(x)}{\proba(x) - n_t^{\epsilon - 1/2}})^{\pmax - 1} \potent(\stratbeq)
 \end{equation*}

We conclude this proof by computing the value of the potential of the complete information game with the linear regression game equilibrium:
\begin{align*}
\potentci(\stratbeq, X) &= \sum_x \sum_t c_t(\strateq_t(x)) n_t^x + \scalar((\sum_x xx^{\top} \sum_t \strateq_t(x) n_t^x)^{-1})\\
& = \sum_x \sum_t c_t(\strateq_t(x)) \frac{n_t^x}{\proba(x)n_t} n_t \proba(x) + \scalar((\sum_x xx^{\top} \sum_t \frac{n_t^x}{\proba(x)n_t} n_t \proba(x))^{-1}) \\
& \le \max_{x,t}(\frac{n_t^x}{\proba(x)n_t}) \sum_x \sum_c c_t(\strateq_t(x)) n_t \proba(x) + \frac{1}{(\min_{x,t}(\frac{n_t^x}{\proba(x)n_t}))^q}\scalar((\sum_x xx^{\top} \sum_t \strateq_t(x) n_t\proba(x))^{-1})\\
&\le D_n \potent(\stratbeq),
\end{align*}
where $D_n = \max(\max_{x,t}(\frac{\proba(x) + n_t^{\epsilon - 1/2}}{\proba(x)n_t}), \frac{1}{(\min_{x,t}(\frac{\proba(x)}{\proba(x) - n_t^{\epsilon - 1/2}}))^q})$.

Combined with the previous result, we obtain:
\begin{equation*}
\potentci(\stratbeq, X) \le D_n \max_{x, t}(\frac{\proba(x) + n_t^{\epsilon - 1/2}}{\proba(x) })^{\pmax - 1} \potentci(\stratbeqci, X).
\end{equation*}

We can show similarly that:
\begin{equation*}
\potent(\stratbeq_{\ci}) \le D_n^\prime \max_{x, t}(\frac{\proba(x)}{\proba(x) - n_t^{\epsilon - 1/2}})^{\pmax - 1}\potent(\stratbeq),
\end{equation*}
where $D_n^\prime = \max(\max_{x,t}(\frac{\proba(x)}{\proba(x) - n_t^{\epsilon - 1/2}}), \frac{1}{(\min_{x,t}(\frac{\proba(x)}{\proba(x) + n_t^{\epsilon - 1/2}}))^q})$.

\end{proof}

\section{Ordinary least squares}
\label{sec:ols}
%

In this section we present the model where the analyst uses the OLS estimator instead of the GLS estimator. We show that, while the use of the OLS estimator removes a strong assumption of our model (the knowledge of the variance of the data points by the analyst), the use of OLS might also highly degrade the estimation cost when agents are not identical. We show in particular that for any game using the OLS estimator, a single agent participating with prohibitively high provision cost can ruin the estimation.

Let us first define the strategic linear regression in the OLS setting. Formally, the analyst receives $\nbplay$ couples $(x_i,\hat{y}_i)$ and uses them to produce an estimate $\betabh$ that is then sent to the agents. Note that we do not assume in this setting that the analyst receives the precision associated to the data points as it is not needed for the estimation.
In what follows, we assume that the analyst computes this estimate by using \emph{ordinary least squares} (OLS) and we denote it by $\betabh_\ols$. OLS is the least squares regression which is optimal in the case of homoskedastic data. It is however sub-optimal when  data are heteroskedastic but still applicable. It is one of the most widespread estimators in general, in particular because, unlike GLS, it is easy to apply and does not require knowledge of the variance of the data points. The covariance of OLS is independent of $\hat{y}_i$ and is equal to $\p{\sum_{i \in N} x_ix_i^{\top}}^{-1}\sum_{i\in N} \frac{x_ix_i^{\top}}{\lambda_i(x_i)}\p{\sum_{i \in N} x_ix_i^{\top}}^{-1}$.
Note that this quantity is well defined only if each $\lambda_i(x_i)$ is strictly positive, unlike GLS that only requires the information matrix ($\sum_i \lambda_i(x_i) x_ix_i^{\top}$) to be invertible.

In a system where data point $\hat{y}_i$ is revealed with precision $\ell_i$, the covariance of $\betabh_\ols$ is
\begin{align*}
  \p{\sum_{i \in N} x_ix_i^{\top}}^{-1}\sum_{i\in N} \frac{x_ix_i^{\top}}{\ell_i}\p{\sum_{i \in N} x_ix_i^{\top}}^{-1}
\end{align*}
In our model, the values of $x_i$ are generated randomly according to a common underlying distribution $\mu$ on $\Xset$. Hence, we define the OLS estimation cost as
\begin{align}
  \label{eq:Uols}
  \Cestimols(\lambdab) &= \scalar\Bigg(\esp{\p{\sum_{i \in N} x_ix_i^{\top}}^{-1}\sum_{i\in N} \frac{x_ix_i^{\top}}{\lambda_i(x_i)}\p{\sum_{i \in N} x_ix_i^{\top}}^{-1}}\Bigg).
\end{align}

We denote $\Gamma_\ols$ (resp. $\Gamma_\gls$) an instance of the game where the analyst uses the OLS (resp. GLS) estimator. For a given precision profile, we define $\potent_\ols(\lambda_i,\lambdab_{-i})$ 
\begin{equation}
  \label{avg.potential.ols}
  \potent_\ols(\stratb) = \sum_{j = 1}^\nbplay
  \esp{\cost_j(\strat_j(x))} + \Cestimols(\stratb).
\end{equation}

We show that our main results still holds in this setting.

\begin{proposition}
Under Assumptions~\ref{structure.information}, \ref{scalarization.type}, and \ref{privacy.type}, a precision profile $\lambdabeq$ is a Nash equilibrium of the OLS linear regression game if and only if it minimizes $\potent_\ols$.
 Such an equilibrium exists.
It is unique if all provision cost functions $c_i$ are strictly convex.
 When there are multiple equilibria, the estimation cost $\Cestimols(\lambdabeq)$ does not depend on the equilibrium.
\end{proposition}

The proof of this proposition is a trivial adaptation of Proof~\ref{proof.uniqueness}. The game $\Gamma_\ols$ thus has the same basic properties as $\Gamma_\gls$ and we can now state our main result in this new model:

\begin{proposition}
 Let $\Gamma_\ols$ be a game satisfying the Assumptions of Theorem~\ref{thm.conv}. Then, with the same constants $\constmin, \constmax > 0$ as Theorem~\ref{thm.conv} that do not depend on $\nbplay$, we have that:
  \begin{equation}
    \label{eq.conv.ols}
  \constmin \nbplay^{-\qhomo\frac{\pmin - 1}{\pmin + \qhomo}-\alpha} \le \Cestimols(\lambdabeq)  \le  \constmax\nbplay^{-\qhomo\frac{\pmin - 1}{\pmin + \qhomo}},
\end{equation}
where $\alpha=\qhomo\frac{(\pmax-\pmin)(q+1)}{\pmax(q+\pmin)}$.
\end{proposition}

\begin{proof}
 \noindent \textbf{Upper bound}

 The proof of the upper bound is the same as in the proof of Theorem~\ref{thm.conv}. Indeed,
 we compute the value of the potential function for a particular constant strategy in which all players use the precision $\lambda(x)=\nbplay^{-\frac{\qhomo + 1}{\pmin + \qhomo}}$ for all values of $x\in\Xset$. It is then sufficient to observe that for such a strategy, we have homoskedasticity of the data points. Thus, the GLS estimator and the OLS estimator are the same and the algebra of the proof can trivially be applied.

\noindent \textbf{Lower bound}

It is sufficient to observe that for all $\lambdab$, we have $\esp{\p{\sum_{i \in N} x_ix_i^{\top}}^{-1}\sum_{i\in N} \frac{x_ix_i^{\top}}{\lambda_i(x_i)}\p{\sum_{i \in N} x_ix_i^{\top}}^{-1}} \succeq \p{\esp{\sum_{i\in N} \lambda_i(x_i) x_ix_i^{\top}}}^{-1}$ by Aitken's theorem of optimality of GLS.

\end{proof}

\subsection*{Differences between the asymptotic behavior of GLS and OLS}

In this section, we show that, while our main result holds when the analyst uses the OLS estimator, $\Gamma_\gls$ and $\Gamma_\ols$ behave fundamentally differently when only subsets of agents satisfy our non-trivial assumptions.

\begin{proposition}
  \label{thm.conv.sub}
  Assume that Assumptions~\ref{structure.information}, \ref{scalarization.type} and \ref{privacy.type} hold. Assume that for all $i \in N$, we have $c_i(0) = 0$. Additionally, assume that there exist $\pmin\ge1$ a function $c_{\max}:\R_+\to\R_+$ and $\subplayset \subseteq N$ such that for all $i\in \subplayset$ and all $a>1, \ell>0$: $a^{\pmin}c_i(\ell) \le c_i(a\ell)$ and $ c_i(\ell)\le c_{\max}(\ell)<\infty$. Then there exists a constant $\constmax > 0$ that does not depend on $|\subplayset|$ and such that:
  \begin{equation}
    \label{eq.conv.sub}
\Cpublic{\lambdabeq} \le  \constmax |\subplayset|^{-\qhomo\frac{\pmin - 1}{\pmin + \qhomo}},
\end{equation}
\end{proposition}

 \begin{proof}

We define the particular constant strategy $$\lambda_i(x)= \left\{
    \begin{array}{ll}
        |\subplayset|^{-\frac{\qhomo + 1}{\pmin + \qhomo}} & \mbox{if } i \in \subplayset \\
        0 & \mbox{Otherwise.}
    \end{array}
    \right.$$
    
The algebra to obtain the bound is then exactly the same as in Section~\ref{proof.conv}.
 
 \end{proof}

This proposition states that for any subset of agents, the convergence rate of the estimation cost is at least as good as if only those agents participated. For example, if half a population suffers from linear provision cost $c_i(\lambda) = \lambda$ while the other half of the population has highly convex provision costs $c_i(\lambda) = \lambda^p$, the estimation cost will converge to $0$ with rate at least $n^{-\qhomo\frac{p - 1}{p + \qhomo}}$. This is significant as we have previously proved that if only agents with linear provision costs participate, GLS is not consistent and the estimation cost does not go to $0$. This property is tightly linked to the GLS estimator. Indeed, GLS weights the data points according to their precision and low precision data points do not hinder the estimation. Formally, for any $\lambdab, \lambda_{n + 1}$, we have $\sum_i \lambda_{i = 1}^{n + 1}(x_i) x_ix_i^{\top} \succeq \sum_{i = 1}^{n} \lambda_i(x_i) x_ix_i^{\top}$ thus adding a data point can only improve the information matrix of the estimator. This is no longer true when using the OLS estimator as it gives the same weight to widely inaccurate data points as to very precise data points.

 We show this difference on an example. We consider an OLS regression game where $n$ agents are willing to give precise data (they have low provision cost) while one agent suffers from prohibitively high provision cost. Formally, let us consider $\Gamma_\ols$ the game where $\Xset = \{1\}$, $n + 1$ agents participate, $c_i(\lambda) = \lambda^p$ for all $i$ in $\{1, \dots, n\}$ and $c_{n + 1}(\lambda) = (n + 1)^2\lambda$. In the following game, we also consider the scalariation $\scalar(\cdot)$ to be the trace which in this case is the identity function.  We have in this game the following potential:
 \begin{equation}
  \potent_\ols(\lambdab) = \sum_{i = 1}^n \lambda_i^p + (n + 1)^2 \lambda_{n + 1} + \frac{1}{(n + 1)^2}\sum_{i = 1}^{n + 1} \frac{1}{\lambda_i}
 \end{equation}
 
 It is then easy to show that at equilibrium, we have $\lambda^*_i = (n + 1)^{-2/(p + 1)}$ for all $i$ in $\{1, \dots, N\}$ and $\lambda^*_{n + 1} = (n + 1)^{-2}$ . This implies that the equilibrium achieves the following estimation cost:
 \begin{equation}
  \label{estim.cost.example}
  \Cestimols(\lambdabeq) = \frac{1}{(n + 1)^2} n (n + 1)^{2/(p + 1)} + 1
 \end{equation}

This estimation cost does not converge to $0$ when $n + 1$ grows large. Also note that even if $p$ grows large meaning that $n$ of the $n + 1$ agents almost do not suffer any cost for providing data, the estimation cost still does not converge to $0$. In contrast, 
the cost functions we defined satisfy the assumptions of Proposition~\ref{thm.conv.sub} meaning that if the analyst used the GLS estimator, they would obtain a consistent estimator with convergence rate at least $n^{-\qhomo\frac{p - 1}{p + \qhomo}}$. Alternatively, if the analyst refused the participation of agent $n + 1$, they would also obtain a consistent estimator. This implies that designing a mechanism to control participants in the OLS model could greatly improve the estimation cost at equilibrium in some cases. This remains an open problem.

\section{Extension to joint distributions}
\label{sec:joint}
In this section, we show how our main result can be extended to a setting where the data points $x_i$ of agents are not independent and identically distributed but are distributed according to a joint distribution $\muj$. We make the following assumption on this joint distribution to ensure the non-triviality of the game:
\begin{assumption}
  \label{structure.information.2}
  The set $\Xset$ is finite and $\espj{\sum_{i \in N}x_ix_i^{\top}}$ is positive definite.
\end{assumption}
For the rest of this section, we omit the subscript denoting that the expected value is taken with regard to the distribution $\muj$.

Having a joint distribution does not change the basic structure of the game. The game is still a potential game with potential
\begin{equation*}
  \potent(\stratb) = \esp{\sum_{j = 1}^\nbplay
  \cost_j(\strat_j(x_j))} + \Cpublic{\stratb}.
\end{equation*}
Note that we still assume that each agent strategy is a function $\lambda_i:\Xset\to \R_+$ for ease of notation. We do not assume, however, that each agent holds each vector with non-zero probability. This implies that if an equilibrium exists, there exists an infinite number of equilibrium as agents may freely choose the precision of the data points that hold with probability zero (without changing their payoffs). As these precision are a simple modeling artifact without any impact on payoffs, we set them to $0$ by convention.

Also note that there may now exist Nash equilibria $\lambdabeq$ for which $\Cpublic{\lambdabeq}=\infty$. For instance, if $\ndim\ge2$ and the joint distribution is such that $\muj(\xb) = 1$ for some $\xb = (x_1, \dots, x_n)$, then $\lambdabeq = 0$ is a Nash equilibrium. Indeed, in that case, no agent has an incentive to deviate since a single agent deviation still yields a non-invertible information matrix (recall that the information matrix is $\esp{\sum_{i\in N} \lambda_i(x_i) x_ix_i^{\top}}$ and that the covariance is the inverse of this matrix). More generally, any profile $\lambdab$ such that the information matrix is non-invertible and remains non-invertible under unilateral deviations is an equilibrium. 
Following \cite{Gast20a}, we call Nash equilibria at which the estimation cost is infinite ``trivial equilibria.'' These are not our focus as they can be avoided with model adjustments such as having $d$ non-strategic agents with data points spanning $\R^d$ guaranteeing a finite covariance.

We claim that Proposition~\ref{thm.uniqueness} (which states that the game has at least one equilibrium, and that if there are multiple equilibria, they have the same estimation cost) still holds for non-trivial equilibria under the extending model where the data points $x_i$ are distributed according to a joint distribution $\muj$ satisfying Assumption~\ref{structure.information.2}; with the following adapted proof. Note that the first step of this version of the proof is inspired from \cite{Gast20a} to handle trivial equilibrium.
\begin{proof}

\textbf{Step 1: The potential function is convex.} The potential function
$\potent(\lambdab)= \esp{\sum_{j = 1}^\nbplay c_j(\strat_j(x_j))} + \Cpublic{\stratb}$ takes values in the
extended positive real numbers line
$\bar{\R}_+=\R_+\cup\{+\infty\}$.

Recall that $\Cpublic{\stratb} = \scalar\Bigg(\p{\esp{\sum_{i\in N} \lambda_i(x_i) x_ix_i^{\top}}}^{-1}\Bigg)$. We denote $V(\lambdab) = \esp{\sum_{i\in N} \lambda_i(x_i) x_ix_i^{\top}}^{-1}$ and $M(\lambdab) = \esp{\sum_{i\in N} \lambda_i(x_i) x_ix_i^{\top}}$. We have that $V(\lambdab)$ is
strictly convex and goes to infinity when $M(\lambdab)$
goes to a non-invertible matrix (i.e., the largest eigenvalue of $V$ goes to infinity for any sequence $\lambdab_n$ that converges to a $\lambdab$ such that $M(\lambdab)$ is non-invertible). As $F$ is convex and increasing, this
shows that $\Cpublic{\lambdab}$ is strictly convex and goes to
$+\infty$ when $M(\lambdab)$ goes to a non-invertible matrix, which then 
implies that $\Cpublic{\lambdab}:\R_+^n\to\bar{\R}_+$ is
continuous. As the functions $c_i$ are convex, we conclude that the
potential function $\potent$ is strictly convex and continuous on
$\bar{\R}_+$.

\textbf{Step 2: The potential admits a minimum.} We first consider the potential evaluated at an arbitrary value and show that this implies boundedness of agents precision at equilibrium. Let $\potent(\mathbf{1}) = \esp{\sum_i c_i(1)} + \scalar((\esp{\sum_i x_ix_i^T})^{-1})$. By Assumption~\ref{privacy.type}, $\lim_{\ell \to +\infty}\cost_i(\ell) = +\infty$. For all $x \in \Xset$, we denote $\mu_i(x)$ the the probability that agent $i$ has data point $x$ when data points are generated with the joint distribution $\muj$. If $\mu_i(x) = 0$, then the value of $\strat_i(x)$ does not change the potential and we can set it to $0$. Otherwise, $\lim_{\ell \to +\infty}\cost_i(\ell)\mu(x) = +\infty$. Hence, there exists $\ellmax$ such that for all $i$ and all $x$, $c_i(\ellmax)\mu(x) > \potent(\mathbf{1})$. This shows that if $\lambdab$ is a precision profile such that $\lambda_i(x)>\ellmax$ for some $i$ and $x$, then $\potent(\lambdab)\ge\potent(\mathbf{1})$.

Let $B$ be the subset of $\lambdab$ such that
$\potent(\lambdab)\le\potent(\mathbf{1})$. By continuity and
convexity of $\potent$, $B$ is a non-empty convex and compact subset of
$[0,\ellmax]^n$ on which $\potent(\lambdab)<\infty$. This implies that there $\potent$ admits a minimum and that
all global minimum of $\potent$ are attained in $B$.

\textbf{If different non-trivial equilibria exist, they have the same estimation cost.} This step is strictly the same as the proof found in Section~\ref{proof.uniqueness}.
\end{proof}

We are now ready to state our main result adapted to this setting. In the following theorem, $\lambdabeq$ denotes any non-trivial equilibrium.

\begin{theorem}
  \label{thm.conv.2}
  Assume that Assumptions~\ref{scalarization.type}, \ref{privacy.type}, and \ref{structure.information.2} hold. Additionally, assume that there exist $\pmin,\pmax\ge1$ and functions $c_{\min},c_{\max}:\R_+\to\R_+$ such that for all $i\in N$ and all $a>1, \ell>0$: $a^{\pmin}c_i(\ell) \le c_i(a\ell) \le a^{\pmax}c_i(\ell)$ and $0<c_{\min}(\ell)\le c_i(\ell)\le c_{\max}(\ell)<\infty$. Then there exist constants $\constmin^\prime, \constmax^\prime > 0$ that depend on $\nbplay$ only through $\espj{\frac{1}{n}\sum_{i \in N}x_ix_i^{\top}}$ and such that:
  \begin{equation}
    \label{eq.conv.2}
  \constmin^\prime \nbplay^{-\qhomo\frac{\pmin - 1}{\pmin + \qhomo}-\alpha} \le \Cpublic{\lambdabeq} \le  \constmax^\prime\nbplay^{-\qhomo\frac{\pmin - 1}{\pmin + \qhomo}},
\end{equation}
where $\alpha=\qhomo\frac{(\pmax-\pmin)(q+1)}{\pmax(q+\pmin)}.$
\end{theorem}

\begin{proof}
In this first step, we compute the value of the potential function for a particular constant strategy in which all players use the precision $\lambda(x)=\nbplay^{-\frac{\qhomo + 1}{\pmin + \qhomo}}$ for all values of $x\in\Xset$.  By abuse of notation, we denote this precision profile by  $(\nbplay^{-\frac{\qhomo + 1}{\pmin + \qhomo}},\dots,\nbplay^{-\frac{\qhomo + 1}{\pmin + \qhomo}})$. The value of the potential for this precision profile is
\begin{align}
  \potent(\nbplay^{-\frac{\qhomo + 1}{\pmin +
  \qhomo}},\dots,\nbplay^{-\frac{\qhomo + 1}{\pmin + \qhomo}})
  &= \esp{\sum_{\play = 1}^{\nbplay} \cost_\play(\nbplay^{-\frac{\qhomo + 1}{\pmin + \qhomo}})} + \scalar((\esp{\sum_{\play = 1}^{\nbplay} x_ix_i^T \nbplay^{-\frac{\qhomo + 1}{\pmin + \qhomo}}})^{-1})\nonumber\\
  &=\sum_{\play = 1}^{\nbplay}\cost_\play(\nbplay^{-\frac{\qhomo + 1}{\pmin + \qhomo}})
    +\scalar((n^{\frac{\pmin-1}{\pmin + \qhomo}} \esp{\frac{1}{n}\sum_{i \in N}x_ix_i^{\top}})^{-1})\nonumber\\
  & \le \sum_{\play = 1}^\nbplay \nbplay^{-\pmin\frac{\qhomo + 1}{\pmin + q}} \cost_i(1) + \scalar((\nbplay^{\frac{\pmin-1}{\pmin + q}} \esp{\frac{1}{n}\sum_{i \in N}x_ix_i^{\top}})^{-1})\label{eq:line2.2}\\
  &=\nbplay^{-\pmin\frac{\qhomo + 1}{\pmin + q}} \sum_{\play = 1}^\nbplay \cost_i(1) + \nbplay^{\frac{\qhomo(1 - \pmin)}{\pmin + \qhomo}}\scalar((\esp{\frac{1}{n}\sum_{i \in N}x_ix_i^{\top}})^{-1})\label{eq:line2bis.2}\\
  &\le \nbplay^{-\frac{\qhomo(\pmin-1)}{\pmin + \qhomo}}c_{\max}(1) + \nbplay^{\frac{\qhomo(1 - \pmin)}{\pmin + \qhomo}}\scalar((\esp{\frac{1}{n}\sum_{i \in N}x_ix_i^{\top}})^{-1})\label{eq:line3.2}\\
  &= \nbplay^{-\frac{\qhomo(\pmin-1)}{\pmin + \qhomo}}\p{c_{\max}(1) + \scalar((\esp{\frac{1}{n}\sum_{i \in N}x_ix_i^{\top}})^{-1})}\label{eq:line4.2},
 \end{align}
 where we use that $c_i(1)\ge a^{\pmin}c_i(1/a)$ with $a=n^{\frac{q+1}{\pmin+q}}$ (from the theorem's assumption) in \eqref{eq:line2.2}, the homogeneity of $F$ (Assumption~\ref{scalarization.type}) in \eqref{eq:line2bis.2}, and the theorem's assumption, which implies that $c_i(1)\le c_{\max}(1)$ for all $i$, in \eqref{eq:line3.2}.

 As $c_i(\ell)\ge0$ and $\lambdabeq$ is a minimum of the potential, it
 holds that
 \begin{align*}
   \Cpublic{\lambdabeq}&\le\potent(\lambdabeq)\le\potent(\nbplay^{-\frac{\qhomo + 1}{\pmin +
  \qhomo}},\dots,\nbplay^{-\frac{\qhomo + 1}{\pmin + \qhomo}}).
 \end{align*}
 Hence, the right-hand-side of \eqref{eq.conv.2} holds with $D=\p{c_{\max}(1) + \scalar((\esp{\frac{1}{n}\sum_{i \in N}x_ix_i^{\top}})^{-1})}$. 

\paragraph{Lower bound.} The lower bound is then simply obtained by plugging the new upper bound of the potential to the proof of the lower bound obtained in Section~\ref{proof.conv}.
\end{proof}

The main difference between Theorem~\ref{thm.conv.2} and Theorem~\ref{thm.conv} is that in Theorem~\ref{thm.conv}, the constants $d$ and $D$ \emph{do} not depend $n$ whereas in Theorem~\ref{thm.conv.2}, the constants $d'$ and $D'$ do depend on $\espj{\frac{1}{n}\sum_{i \in N}x_ix_i^{\top}}$. This is because in Theorem~\ref{thm.conv.2}, we do not make any assumption on the joint distribution.  We thus do not have any guarantee that the joint distribution will have some stable property when the number of agents grow.  On the other hand, if $\espj{\frac{1}{n}\sum_{i \in N}x_ix_i^{\top}}$ is independent on $n$, the constants $d'$ and $D'$ will also not depend on $n$. 

In fact, the multiplicative terms of Theorem~\ref{thm.conv.2} are simply obtained by replacing $\esp{xx^T}$ with $\espj{\frac{1}{n}\sum_{i \in N}x_ix_i^{\top}}$ in the multiplicative terms of Theorem~\ref{thm.conv} (note that we retrieve Theorem~\ref{thm.conv} when data points are iid). This latter term captures precisely the impact of correlation on the estimation cost. For instance,if data points are highly correlated in a way that poorly represents the input space, $\scalar((\espj{\frac{1}{n}\sum_{i \in N}x_ix_i^{\top}})^{-1})$ can be arbitrarily large, leading to a commensurately large upper bound (the corresponding lower bound behavior is similar).

\section{Additional illustrations}
\label{sec:more_experiments}
%
%

\newcommand\fig[1]{
  \begin{minipage}{.2\linewidth}
    \includegraphics[width=\linewidth]{figures_paper/#1}
  \end{minipage}
}
\newcommand\figs[1]{
  \fig{#1_opt}&\fig{#1_101}&\fig{#1_120}&\fig{#1_150}\\
}
\newcommand\figdistrib[2]{
  \fig{#2}&\fig{#1_101}&\fig{#1_120}&\fig{#1_150}\\
}
\newcommand\figd[2]{
  $d=#1$&\figs{#2#1}
}

\subsection{Illustration of the equilibrium characterization}

In this section, we provide additional illustrations on the equilibrium characterization (Section~\ref{sec:equilibrium}), which complement Figure~\ref{fig.precs} and show that the discussion on that figure in the paper continues to apply in different settings, namely: 
\begin{enumerate}[a.,leftmargin=22pt,itemsep=0pt,topsep=0pt]
\item In Figure~\ref{fig:various_degree}, we vary the degree $d$ of the polynomial regression (Figure~\ref{fig.precs} has $d=4$). 

\item In Figure~\ref{fig:various_distribution}, we vary the distribution $\mu$ (Figure~\ref{fig.precs} has a uniform distribution that corresponds to the first row in Figure~\ref{fig:various_distribution}). Here, we fix $d=4$ and we do not plot the optimal design as it does not depend on $\mu$. 

\item In  Figure~\ref{fig:various_scalarization}, we use a different scalarization, the squared Frobenius norm ($F(M)=\sum_{ij}M_{ij}^2$), while keeping a uniform distribution $\mu$ and $d=4$. 
\end{enumerate}

\begin{figure}[ht]
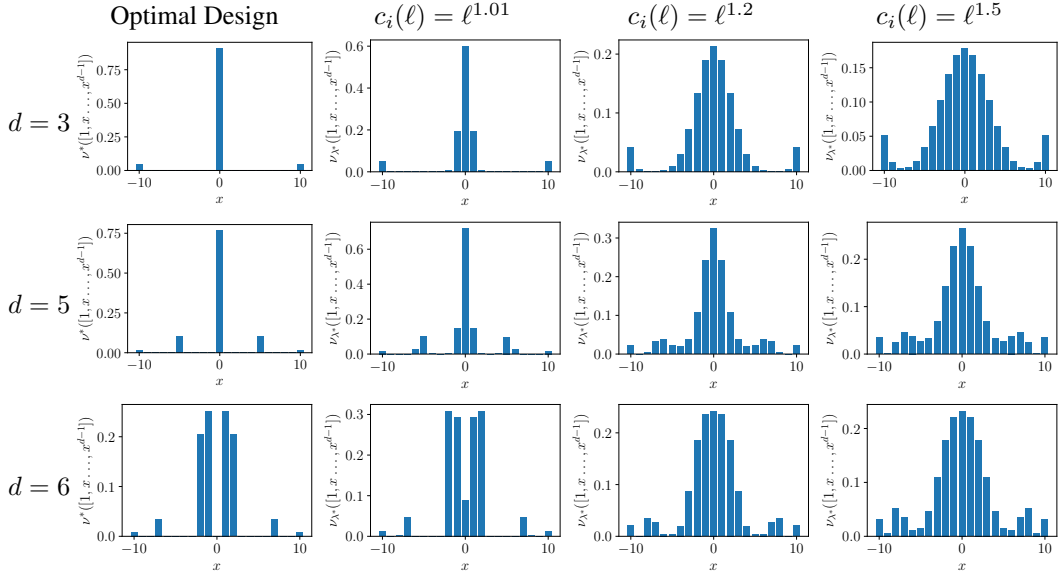

  \centering
  \begin{tabular}{c@{}c@{}c@{}c@{}c}
    &Optimal Design & $c_i(\ell)=\ell^{1.01}$ & $c_i(\ell)=\ell^{1.2}$ & $c_i(\ell)=\ell^{1.5}$\\
    \figd{3}{unif_deg}%
    \figd{5}{unif_deg}%
    \figd{6}{unif_deg}
  \end{tabular}
  \caption{Optimal design $\nu^*$ and allocation of precision at equilibrium $\nu_{\lambdabeq}$ with various degrees $d$ of the polynomial regression (here $\mu$ is uniform and the scalarization is the trace as in Figure~\ref{fig.precs}).}
  \label{fig:various_degree}
\end{figure}
\FloatBarrier

Figure~\ref{fig:various_degree} illustrates the optimal design $\nu^*$ and the allocation of precision at equilibrium $\nu_{\lambdabeq}$ as defined in Theorem~\ref{thm.optim.design} in the same setting as Figure~\ref{fig.precs} ($d=4$) with different degrees for the polynomial regression ($d=3, 5, 6$). We observe that for $d = 3$ and $d = 5$, the optimal design puts maximal weight on the central vector $[1, x, \cdots, x^{d-1}]$ with $x = 0$ while for $d = 4$ and $d=6$, this vector does not belong to the support of the optimal design. We observe a similar property for the equilibrium of games with near-linear data provision cost. The allocations of precision at equilibrium for $p = 1.2$ and $p=1.5$, however, are significantly different from the optimal design for all values of $d$ (in particular with a maximum of precision for the central vector with $x=0$), and they have a shape that does not significantly vary with  the degree $d$.

\begin{figure}[ht]
  \centering
  \begin{tabular}{c@{}c@{}c@{}c@{}c}
    Distribution $\mu$ & $c_i(\ell)=\ell^{1.01}$ & $c_i(\ell)=\ell^{1.2}$ &
                                                                            $c_i(\ell)=\ell^{1.5}$\\
    \figdistrib{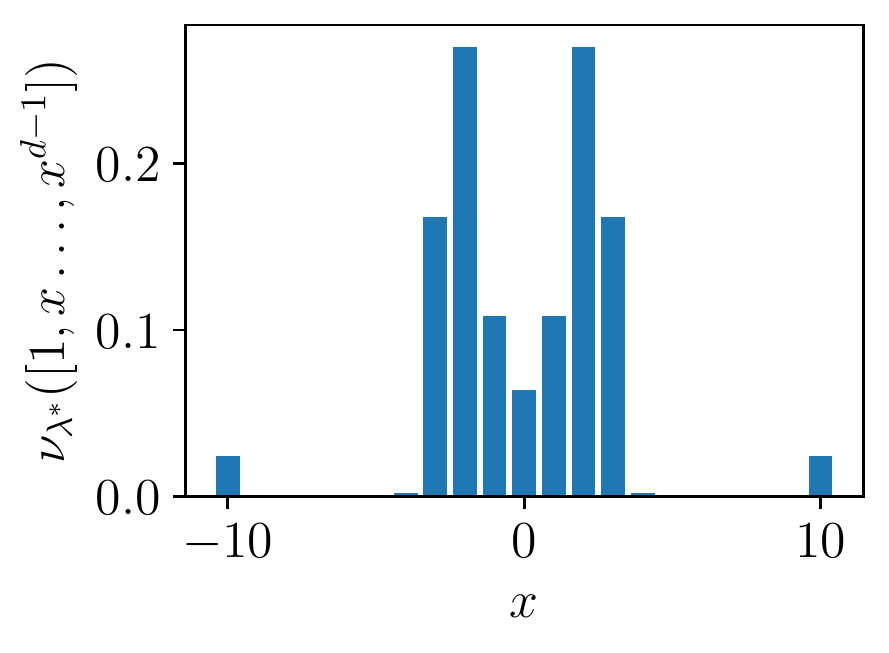}{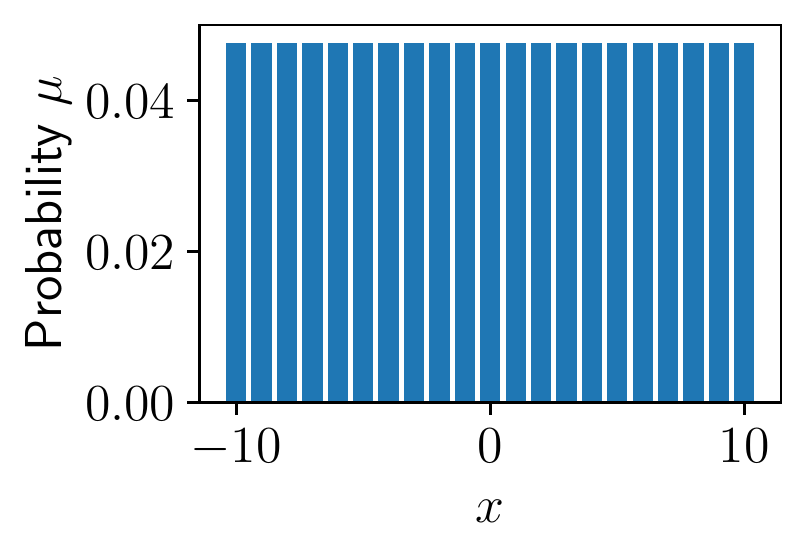}%
    \figdistrib{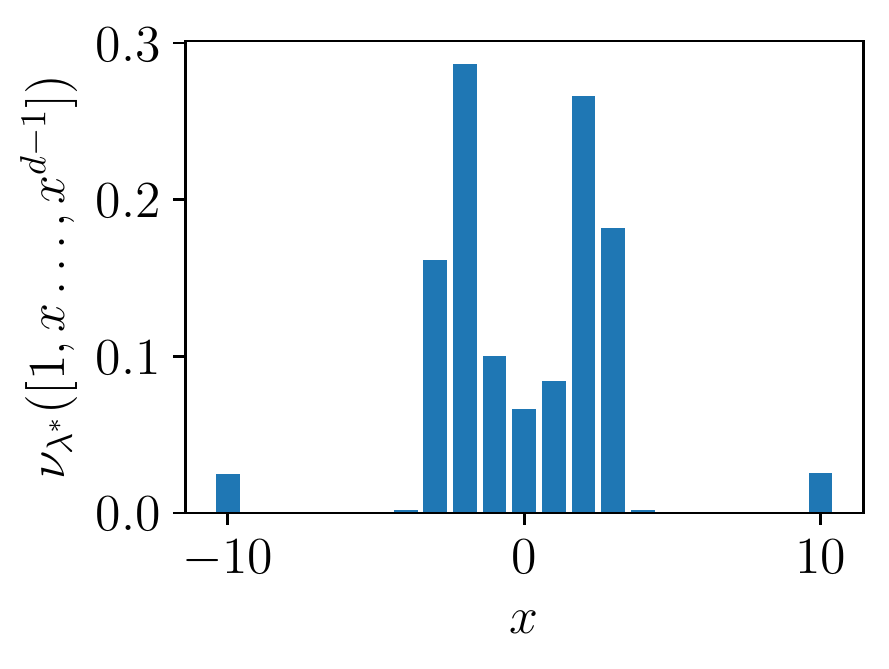}{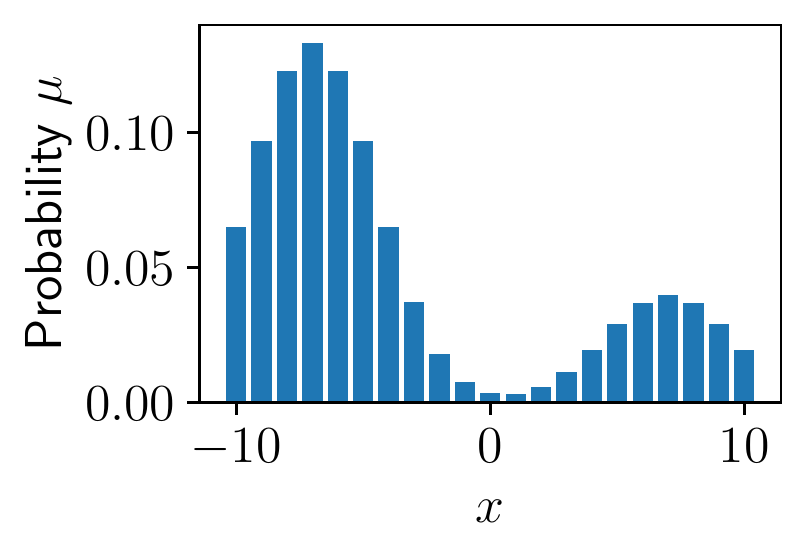}%
    \figdistrib{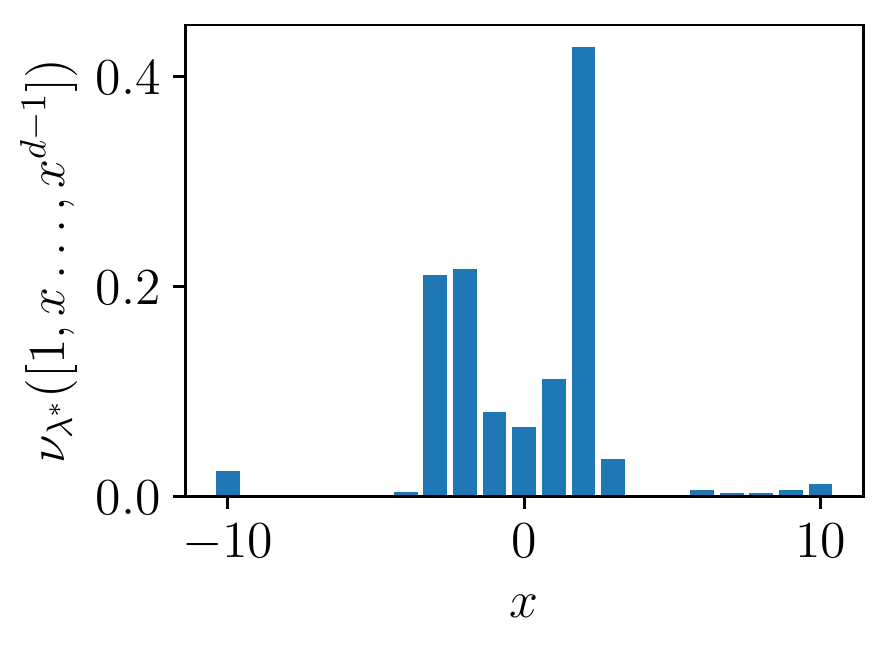}{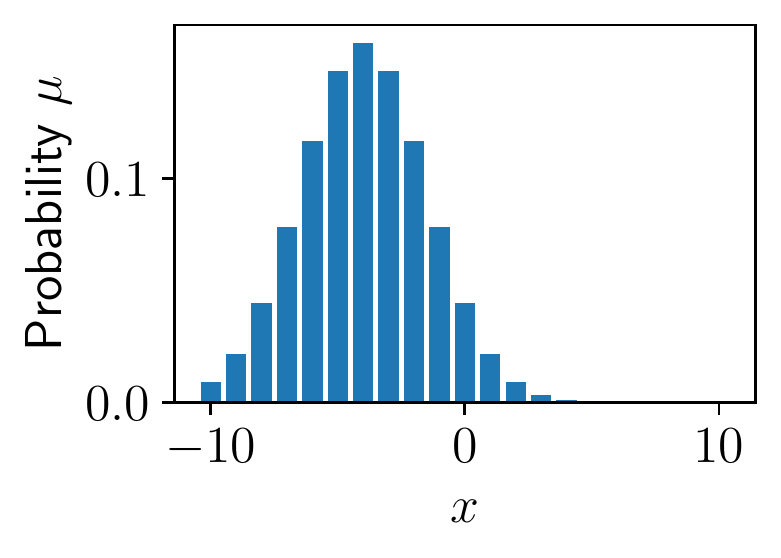}%
  \end{tabular}
  
  \caption{Allocation of precision at equilibrium $\nu_{\lambdabeq}$ with various distributions $\mu$ (here $d=4$ and the scalarization is the trace as in Figure~\ref{fig.precs}). The optimal design $\nu^*$ does not depend on $\mu$ and is therefore the same as in Figure~\ref{fig.precs}.}
  \label{fig:various_distribution}
\end{figure}
\FloatBarrier

Figure~\ref{fig:various_distribution} illustrates the allocation of precision at equilibrium $\nu_{\lambdabeq}$ as defined in Theorem~\ref{thm.optim.design} in the same setting as Figure~\ref{fig.precs} with various distributions $\mu$ of the agents' $x_i$ vectors. The first row of graphs corresponds to the exact same setting as Figure~\ref{fig.precs} (uniform distribution) while the next rows show the results for other distributions. In addition to Figure~\ref{fig.precs}, we plot the results for monomial costs of exponent $p=1.5$, but we do not plot the optimal design $\nu^*$ as it is the same for all distributions (and shown on Figure~\ref{fig.precs}). We first observe that, for all distributions, the allocation of precision at equilibrium is close to the optimal design (and hence almost independent of the distribution) for near-linear provision costs ($p=1.01$). For more convex provision costs however, the allocation of precision at equilibrium varies with $\mu$ in non-trivial ways. In the second row of Figure~\ref{fig:various_distribution} (compared to the first), we observe that $\nu_{\lambdabeq}([1, x, \cdots, x^{d-1}])$ shrinks for values of $x$ close to $0$. This is explained by two factors: i) vectors with $x$ close to $0$ have a low probability according to $\mu$ and ii) provision costs are superlinear meaning that the agent cannot compensate this probability by multiplying the precision attributed to this vector without prohibitively increasing its cost. We observe a similar behavior for the third row of Figure~\ref{fig:various_distribution} where $\nu_{\lambdabeq}$ has a shape similar to the first row with values skewed to the left where vectors have higher probability.

\begin{figure}[ht]
  \centering
  \begin{tabular}{ccccc}
    &Optimal design & $c_i(\ell)=\ell^{1.01}$ & $c_i(\ell)=\ell^{1.2}$ &
                                                                         $c_i(\ell)=\ell^{1.5}$\\
    &\figs{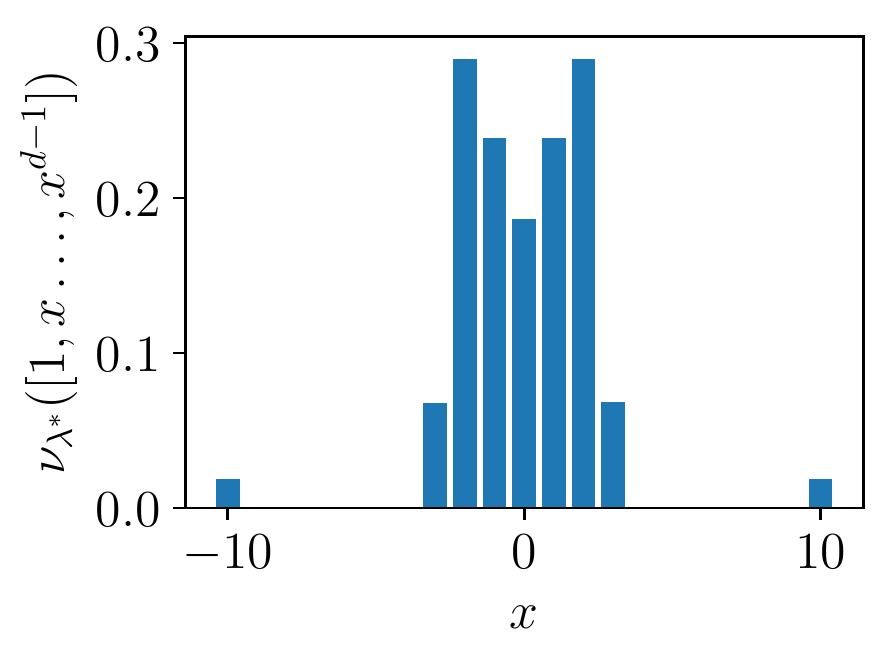}
  \end{tabular}
  
  \caption{Optimal design $\nu^*$ and allocation of precision at equilibrium $\nu_{\lambdabeq}$ with the squared Frobenius norm as a scalarization $F$ (here $\mu$ is uniform and $d=4$ as in Figure~\ref{fig.precs}).}
  \label{fig:various_scalarization}
\end{figure}

Figure~\ref{fig:various_scalarization} illustrates the optimal design $\nu^*$ and the allocation of precision at equilibrium $\nu_{\lambdabeq}$ as defined in Theorem~\ref{thm.optim.design} in the same setting as Figure~\ref{fig.precs} but when using the squared Frobenius norm as a scalarization to define the estimation cost instead of the trace. We observe that both figures show similar trends. In particular, Figure~\ref{fig:various_scalarization} with the squared Frobenius norm exhibits the same behaviors as discussed before on Figure~\ref{fig.precs} for the trace: the allocation of precision at equilibrium is close to the optimal design for $p=1.01$ while it departs significantly for $p=1.2$ and $p=1.5$ where the precision for the vector $[1,0,\dots,0]$ is maximal (instead of zero in the optimal design).

\subsection{Numerical exploration of Theorem~\ref{thm.conv}}

In this section, we explore the result of Theorem~\ref{thm.conv} through numerical simulations. We consider a one-dimensional model with $\Xset = \{1\}$. The scalarization is the trace (which satisfies Assumption~\ref{scalarization.type} with $q = 1$). This means that $\Cestim(\lambdab)=(\sum_i\lambda_i(1))^{-1}$. Recall that Theorem~\ref{thm.conv} shows that 
\begin{align*}
    d n^{-q\frac{\pmin - 1}{\pmin + 1} - \alpha} \le \Cestim(\lambdabeq)  \le 
    D n^{-q \frac{\pmin - 1}{\pmin + 1}}. 
\end{align*}
The goal of this section is to compare the upper and lower bounds of Theorem~\ref{thm.conv} to $\Cestim(\lambdabeq)$, to see if the true convergence rate is close to the lower or to the upper bound.

In the remaining of this subsection, we will display $\Cestim(\lambdabeq)$ as a function of $n$ in loglog-scale and compare it to three possible convergence rates: 
\begin{enumerate}[(a)]
    \item $n^{-q\frac{\pmin - 1}{\pmin + q} - \alpha}$ (the rate of the lower bound of Theorem~\ref{thm.conv});
    \item $n^{-q\frac{\pmin - 1}{\pmin + q}}$ (the rate of the upper bound of Theorem~\ref{thm.conv}, which is the convergence rate when all players have cost $c_i(\ell)=\ell^{\pmin}$);
    \item $n^{-q\frac{\pmax - 1}{\pmax + q}}$ (the convergence rate when all players have cost $c_i(\ell)=\ell^{\pmax}$).
\end{enumerate}
Note that (a) is the fastest convergence rate, followed by (c) and then by (b).

In all plots in this section, we normalize the values such that they all start at the same point for $n=3$ ($n=3$ is the smallest game for which we compute $\Cestim(\lambdabeq)$).



\subsubsection{Heterogeneous agents with different exponents}

We first consider heterogeneous agents. For a given $n$, $n/3$ agents have provision costs $c_i(\ell) = \ell^\pmax$ and $2n/3$ agents have provision costs $c_i(\ell) = \ell^\pmin$. This setup satisfies the assumptions of Theorem~\ref{thm.conv} with the corresponding $\pmin$ and $\pmax$. We consider two setups: $(\pmin,\pmax) = (1, 4)$ and $(\pmin, \pmax) = (2, 3)$.

Figure~\ref{fig.heterogeneous} compares the convergence rate of $\Cestim(\lambdabeq)$ to the three bounds defined above. This figure suggests that the estimation cost behaves as when all players have estimation cost $\ell^{\pmax}$. Intuitively, this is explained by the fact that in the game, an agent that has a cost $c_i(\ell)=\ell^\pmin$ will give a very small precision. Hence, the game will almost behave as if this agent was not in the game. This explains why the convergence rate of $\Cestim(\lambdabeq)$ is driven by agents having exponent $\pmax$.


\begin{figure}
    \centering
    \begin{subfigure}[b]{0.49\textwidth}
     \centering
    \includegraphics[width=\textwidth]{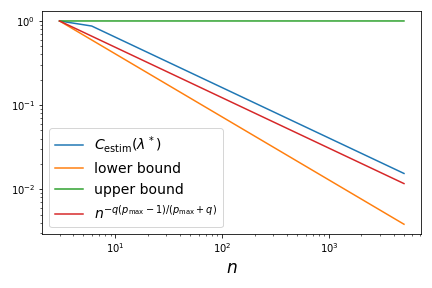}
    \caption{Comparison for $\pmin = 1$ and $\pmax=4$}
    \label{fig.thm.heterogeneous.log1}
    \end{subfigure}
      \begin{subfigure}[b]{0.49\textwidth}
     \centering
    \includegraphics[width=\textwidth]{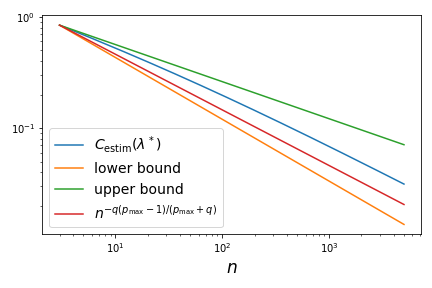}
    \caption{Comparison for $\pmin=2$ and $\pmax=3$}
    \label{fig.thm.heterogeneous.log2}
    \end{subfigure}
    \caption{Comparison of the rate of convergence of the estimation cost with different bounds for agents with heterogeneous costs}
    \label{fig.heterogeneous}
\end{figure}


\subsubsection{Agents with polynomial provision costs}

We then consider agents with polynomial provision costs. We assume that the $n$ agents have the same provision costs $c_i(\ell) = \sum_{k = \pmin}^\pmax \ell^k$. Again, these provision cost satisfy the assumptions of Theorem~\ref{thm.conv} with the corresponding $\pmin$ and $\pmax$.

Figure~\ref{fig.poly} compares the convergence rate of the covariance to the upper and lower bounds of Theorem~\ref{thm.conv}.  We observe that the convergence rate is close to the upper bound $n^{(\pmin - 1)/(\pmin + 1)}$. This result is natural as polynomials are sums of monomials and it is logical to expect the convergence rate to be according to the "worst" monomial of degree $\pmin$.

\begin{figure}
    \centering
    \begin{subfigure}[b]{0.49\textwidth}
     \centering
    \includegraphics[width=\textwidth]{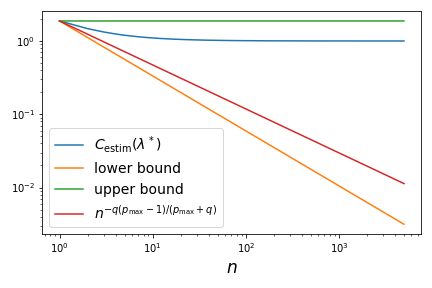}
    \caption{Comparison for $\pmin = 1$ and $\pmax=4$}
    \label{fig.thm.heterogeneous.polylog1}
    \end{subfigure}
      \begin{subfigure}[b]{0.49\textwidth}
     \centering
    \includegraphics[width=\textwidth]{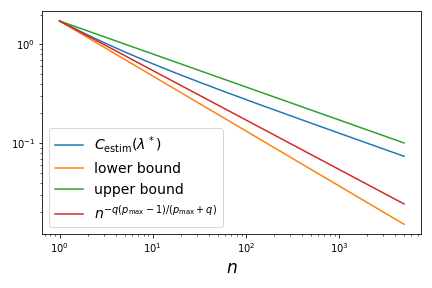}
    \caption{Comparison for $\pmin = 2$ and $\pmax = 3$}
    \label{fig.thm.heterogeneous.polylog2}
    \end{subfigure}
    \caption{Comparison of the rate of convergence of the estimation cost with different bounds for agents with polynomial costs}
    \label{fig.poly}
\end{figure}


\subsubsection{Agents with non-polynomial provision costs}

This result on polynomial functions alongside the fact that the precision of each agent goes to $0$ when the number of agents goes to infinity hints at the behavior of the estimation cost with more general provision costs. Indeed, if agents have provision cost which have a Taylor expansion at $0$, their cost can be well approximated by a polynomial function. The previous figure then suggests that the convergence rate in this case is driven by the first non-null term of the Taylor expansion of the function of degree $\pmin$.

We illustrate this in Figure~\ref{fig.cosh} where we consider homogeneous agents with provision costs $c_i(\ell) = \cosh(\ell) - 1$. Recall that $\cosh(\ell) - 1 = \sum_{k = 1}^\infty \frac{\ell^{2k}}{(2k)!}$. This model therefore satisfy our assumptions with $\pmin = 2$ and $\pmax = \infty$. According to our previous observations, we expect the convergence rate in this case to be the upper bound $(\pmin - 1)/(\pmin + 1)$ with $\pmin = 2$. Note that in this case our lower bound and $n^{-q(\pmax - 1)/(\pmax + 1)}$ both represent convergence rates of $n^{-q}$ corresponding to the non strategic setting.  Figure~\ref{fig.cosh} suggests indeed that the convergence rate is close to this upper bound.

\begin{figure}
    \centering
    \includegraphics[width=0.49\textwidth]{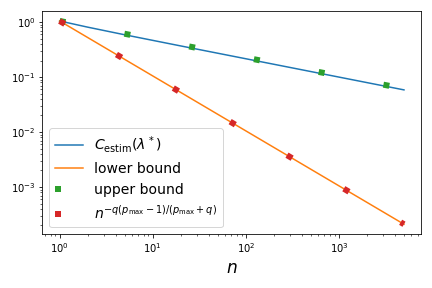}
    \caption{Comparison of the rate of convergence of the estimation cost with the upper bound of Theorem~\ref{thm.conv} for agents with hyperbolic cosine costs.}
    \label{fig.cosh}
\end{figure}

\section{Hardware and software used for experiments}
\label{app.reproducibility}

All experiments were run on a Dell xps-13 laptop with a Quad core Intel Core i7-8550U (-MT-MCP-) CPU under Ubuntu 18.04.
Experiments were made using Python 3 code which was submitted as supplementary material and  will be made publicly available with the full release of the paper. The main libraries used are presented in README.md and the versions used for the experiments are available in requirements.txt.

\end{document}